\documentclass[11pt]{article}

\usepackage[letterpaper,margin=0.89in]{geometry}
\usepackage{amsmath, amssymb, amsthm, amsfonts}

\usepackage{cite}
\usepackage{framed}
\usepackage[framemethod=tikz]{mdframed}
\usepackage{appendix}
\usepackage{graphicx}
\usepackage{color}
\usepackage{algorithm}
\usepackage[noend]{algpseudocode}
\usepackage[textsize=tiny]{todonotes}

\usepackage{enumitem}
\setitemize{noitemsep,topsep=3pt,parsep=3pt,partopsep=3pt}

\definecolor{darkgreen}{rgb}{0,0.5,0}
\usepackage{hyperref}
\hypersetup{
    unicode=false,          
    colorlinks=true,        
    linkcolor=red,          
    citecolor=darkgreen,        
    filecolor=magenta,      
    urlcolor=cyan           
}

\usepackage[nameinlink]{cleveref}

\theoremstyle{plain}
\newtheorem{lemma}{Lemma}[section]
\newtheorem*{lemma*}{Lemma}
\newtheorem{theorem}[lemma]{Theorem}

\newtheorem*{corollary*}{Corollary}

\theoremstyle{definition}
\newtheorem{definition}[lemma]{Definition}

\algnewcommand\algorithmicswitch{\textbf{switch}}
\algnewcommand\algorithmiccase{\textbf{case}}

\algdef{SE}[SWITCH]{Switch}{EndSwitch}[1]{\algorithmicswitch\ #1\ \algorithmicdo}{\algorithmicend\ \algorithmicswitch}%
\algdef{SE}[CASE]{Case}{EndCase}[1]{\algorithmiccase\ #1}{\algorithmicend\ \algorithmiccase}%
\algtext*{EndSwitch}%
\algtext*{EndCase}%

\newcommand{\FullOrShort}{short}

\ifthenelse{\equal{\FullOrShort}{full}}{
	
  \newcommand{\fullOnly}[1]{#1}
  \newcommand{\shortOnly}[1]{}
  \newcommand{\algorithmsize}{\small}
  }{

    \newcommand{\shortOnly}[1]{#1}
		\newcommand{\fullOnly}[1]{}
    \newcommand{\algorithmsize}{\footnotesize}
  }

\newcommand\supercluster{SC}
\newcommand\superclustering{\mathcal{SC}}
\newcommand\Zclustering{\mathcal{C}^0}

\newcommand{\local}{${\mathsf{LOCAL}}$}
\newcommand{\congest}{${\mathsf{CONGEST}}$}
\newif\ifrandom
\randomtrue

\newcommand{\dist}{\mbox{\tt dist}}

\newcommand{\appdeg}{{\widehat{\deg}}}

\newcommand{\diam}{{\rm diam}}
\renewcommand{\paragraph}[1]{\vspace{0.15cm}\noindent {\bf #1}}
\def\dnsparagraph#1{\par\vspace{2pt}\noindent{\bf #1}.}
\def\inline#1:{\par\vskip 7pt\noindent{\bf #1:}\hskip 10pt}

\def\blackslug{\hbox{\hskip 1pt \vrule width 4pt height 8pt
    depth 1.5pt \hskip 1pt}}
\def\QED{\quad\blackslug\lower 8.5pt\null\par}

\newcommand{\ConstructZeroSuper}{\mathsf{ConsZeroSuperclustering}}
\newcommand{\BipartiteSpanner}{\mathsf{SparserBipartiteSpanner}}

\newcommand{\NaiveSpanner}{\mathsf{NaiveSpanner}}
\newcommand{\ImprovedSpanner}{\mathsf{ImprovedSpanner}}
\newcommand{\ImprovedThreeSpanner}{\mathsf{Improved3Spanner}}
\newcommand{\BipartiteThreeSpanner}{\mathsf{Bipartite3Spanner}}

\usepackage{titletoc}

\begin{document}

\title{Improved Deterministic Distributed Construction of Spanners}

\author{
 Ofer Grossman \\
  \small MIT \\
  \small ofer.grossman@gmail.com
	\and
	Merav Parter\\
	\small MIT \\
	\small parter@mit.edu
}

\maketitle

\begin{abstract}
Graph spanners are fundamental graph structures with a wide range of applications in distributed networks. We consider a standard synchronous message passing model where in each round $O(\log n)$ bits can be transmitted over every edge (the \congest\ model). 

The state of the art of deterministic distributed spanner constructions suffers from large messages. 
The only exception is the work of Derbel et al. \cite{derbel2010sublinear}, which computes an optimal-sized $(2k-1)$-spanner but uses $O(n^{1-1/k})$ rounds. 

In this paper, we significantly improve this bound. We present a deterministic distributed algorithm that given an unweighted $n$-vertex graph $G = (V, E)$ and a parameter $k > 2$, constructs a $(2k-1)$-spanner with $O(k \cdot n^{1+1/k})$ edges within $O(2^{k} \cdot n^{1/2 - 1/k})$ rounds for every even $k$. For odd $k$, the number of rounds is $O(2^{k} \cdot n^{1/2 - 1/(2k)})$. For the weighted case, we provide the first deterministic construction of a $3$-spanner with $O(n^{3/2})$ edges that uses $O(\log n)$-size messages and $\widetilde{O}(1)$ rounds. If the nodes have IDs in $[1, \Theta(n)]$, then the algorithm works in only $2$ rounds!
%
%
\end{abstract}

\tableofcontents

\section{Introduction \& Related Work}
Graph spanners are fundamental graph structures that are used as a key building block in various communication applications, e.g., routing, synchronizers, broadcasting, distance oracles, and shortest path computations. For this reason, the distributed construction of sparse spanners has been studied extensively \cite{derbel2006fast,baswana2007simple,derbel2007deterministic,DerbelGPV08,derbel2009local,pettie2010distributed,derbel2010sublinear}. The standard setting is a synchronous message passing model where per round each node can send one message to each of its neighbors. Of special interest is the case where the message size is limited to $O(\log n)$ bits, a.k.a. the \congest\ model.

The common objective in distributed computation of spanners is to achieve the best-known existential size-stretch trade-off as fast as possible: It is folklore that for every graph $G=(V, E)$, there exists a $(2k-1)$-spanner $H \subseteq G$ with $O(n^{1+1/k})$ edges. Moreover, this size-stretch tradeoff is believed to be optimal, following the girth conjecture of Erd\H{o}s.

Designing deterministic algorithms for local problems has been receiving a lot of attention since the foundation of the area in 1980’s. Towards the end of this section, we elaborate more on the motivation for studying deterministic algorithms in the distributed setting. 

\paragraph{State of the art for deterministic distributed constructions of spanners:}
Whereas there are efficient randomized constructions for spanners, 
as the reader will soon notice, the state of the art for distributed deterministic spanner constructions suffers from large message sizes: Derbel and Gavoille \cite{derbel2006fast} construct constant stretch spanners with $o(n^2)$ edges and $O(n^{\epsilon})$ rounds for any constant $\epsilon$, using messages of size $O(n)$. Derbel, Gavoille and Peleg improved this result and presented in \cite{derbel2007deterministic}
a construction of an $O(k)$-spanner with $O(k n^{1+1/k})$ edges in $O(\log^{k-1} n)$ rounds. 
This was further improved in the seminal work of Derbel, Gavoille, Peleg, and Viennot \cite{DerbelGPV08}, which provides a deterministic $k$-round algorithm for constructing $(2k-1)$-spanners with optimal size. However, again the algorithm uses messages of size $O(n)$. 
Using large messages is indeed inherent to all known efficient deterministic techniques, which are mostly based on network decomposition and graph partitioning. In the conventional approaches of network decomposition, the deterministic algorithms for spanners usually require a vertex to learn the graph topology induced by its $O(1)$-neighborhood. This cannot be done efficiently with small messages.  

As Pettie~\cite{pettie2010distributed} explicitly noted, \emph{all} these constructions have the disadvantage of using large messages. Derbel et al.~\cite{derbel2009local} also pointed out that constructing sparse spanners deterministically with small message sizes remains open.

\paragraph{The state of the art when using small messages:} 
There are only two exceptions for this story. 
Barenboim et al. \cite{barenboim2016fast} showed a construction of $O(\log^{k-1}n)$ spanner with $O(n^{1+1/k})$ edges in $O(\log^{k-1}n)$ rounds.  Hence, whereas the runtime is polylogarithmic, the stretch-size tradeoff of the output spanner is quite far from the optimal one. 

We are then left with only \emph{one} previous work that fits our setting, due to Derbel, Mosbah and Zemmari
\cite{derbel2010sublinear}. They provide a  deterministic
construction of an optimal-size $(2k-1)$-spanner but using $O(n^{1-1/k})$ rounds.

\paragraph{The state of the art in other distributed settings:} 
Turning to \emph{randomized} constructions, perhaps one of the most well known approaches to construct a spanner is given by Baswana and Sen \cite{baswana2007simple}, which we review soon. 
Recently, \cite{Censor-HillelPS16} showed that the Baswana-Sen algorithm can be derandomized in the congested clique model of communication in which every pair of nodes (even non-neighbors in the input graph) can exchange $O(\log n)$ bits per round. Note that this model is much stronger than the standard model in which only \emph{neighboring} vertices can communicate. Indeed the algorithm of \cite{Censor-HillelPS16} requires a global evaluation of the random seed, thus implementing this algorithm in the standard \congest\ model requires $\Omega(\diam(G)+n^{1-1/k})$ rounds where $\diam(G)$ is the diameter of the graph. Hence, deterministic construction of spanners in the \congest\ model calls for new ideas!

Before we proceed with introducing our main contribution, we make a short pause to further motivate the study of deterministic distributed algorithms.

\paragraph{A note on deterministic distributed algorithms:}
Much effort has been invested in designing deterministic distributed algorithms for local problems. Examples include MIS (maximal independent set), vertex coloring, edge coloring and and matching. Until recently, a deterministic poly-log $n$ round algorithm was known only for the maximal matching problem (see \cite{hanckowiak2001distributed} and a recent improvement by \cite{fischer2017deterministic}). In a recent breakthrough \cite{fischer2017edgecolor}, a polylogarithmic solution was provided also for the $(2\Delta-1)$ edge coloring\footnote{Where $\Delta$ is the maximum degree in the graph.}.  Aside from the general theoretical question, the distributed setting adds additional motivation for studying deterministic algorithms (as nicely noted in \cite{fischer2017deterministic}). 
First, in the centralized setting, if the randomized algorithm ends with an error, we can just repeat. In the distributed setting, detecting a global failure requires communicating to a leader, which blows up the runtime by a factor of network diameter. Second, for problems as MIS, \cite{ChangKP16} showed that improving the randomized complexity requires an improvement in the deterministic complexity.

Wheres most results for deterministic local problems are for the \local\ model, which
allows unbounded messages; the size of messages that are sent throughout the computation is a second major attribute of distributed algorithms. It is therefore crucial to study the complexity of local
problems under bandwidth restrictions. Surprisingly, most of the algorithms for local
problems already use only small messages. The problem of spanners is distinguished from these problems, and in fact, spanners is the only setting we are aware of, in which all existing deterministic algorithms use \emph{large} messages. Hence, the main challenge here is in the \emph{combination} of \underline{deterministic} algorithms with \underline{congestion} constraints. 
\vspace{-10pt}
\subsection{Our Contribution}

Our main result is:
\begin{mdframed}[hidealllines=false,backgroundcolor=gray!30]
\begin{theorem}
\label{thm:spannerall}
For every $n$-vertex unweighted graph $G=(V,E)$ and even $k$, there exists a deterministic distributed algorithm that constructs a $(2k-1)$-spanner with $O(k \cdot n^{1+1/k})$ edges in $O(2^{k} \cdot n^{1/2-1/k})$ rounds using $O(\log n)$-size messages\footnote{For odd $k$, we obtain a similar theorem, but with $O(2^{k} \cdot n^{1/2-1/(2k)})$ rounds.}.
\end{theorem}
\end{mdframed}

A key element in our algorithm is the construction of sparser spanners for unbalanced bipartite graph. 
This construction might become useful in other spanner constructions. 
\begin{mdframed}[hidealllines=false,backgroundcolor=gray!30]
\begin{lemma}[Bipartite Spanners]\label{bipartite}
Let $G = (A \cup B, E)$ be an unweighted bipartite graph, with $|A|\leq |B|$. For even $k \ge 4$, one can construct (in the \congest\ model) a $(2k-1)$ spanner $H$ with $|E(H)| = O(k |A|^{1+2/k}+ |B|)$ edges within $O(|A|^{1 - 2/k})$ rounds\footnote{Hence, yielding an improved edge bound, for $|A|\leq n^{(k+1)/(k+2)}$.}.
\end{lemma}
\end{mdframed}

Turning to weighted graphs, much less in known about the deterministic construction of spanners in the distributed setting. The existing deterministic constructions of optimal-sized $(2k-1)$-spanners (even in the \local\ model) are restricted to \emph{unweighted} graphs, already for $k=2$. If the edge weights are bounded by some number $W$, there is a simple reduction\footnote{Apply the algorithm for unweighted graphs separately for every weight scale $((1+\epsilon)^i, (1+\epsilon)^{i+1}]$.} to the unweighted setting, at the cost of increasing the stretch by a factor of $(1+\epsilon)$ and the size of the spanner by a factor of $\log_{1+\epsilon}W$.
Hence, already in the \local\ model and $k=2$, we only have a $(3+\epsilon)$ spanner with $\widetilde{O}(n^{3/2})$ edges. Whereas our general approach does not support the weighted case directly, our algorithm for $3$-spanners does extend for weighted graphs. Hence, we give here the first deterministic construction with nearly tight tradeoff between the size, stretch and \emph{runtime}. 
\begin{mdframed}[hidealllines=false,backgroundcolor=gray!30]
\begin{theorem}[$3$-Spanners for Weighted Graphs]
\label{thm:3spanner}
For every $n$-vertex weighted graph $G=(V,E)$, there exists a deterministic distributed algorithm that constructs a $3$-spanner with $O(n^{3/2})$ edges in $O(\log  n)$ rounds using $O(\log n)$-size messages. If vertices have IDs in the range of $[1,O(n)]$, it can be done in two rounds.
\end{theorem}
\end{mdframed}
\vspace{-10pt}
\subsection{Our Approach and Key Ideas in a Nutshell}
For the sake of discussion, let $k = O(1)$ throughout this section.

\paragraph{A brief description of the randomized construction by Baswana-Sen.}
A clustering $\mathcal{C}=\{C_1,\ldots, C_\ell\}$ is a collection of vertex disjoint sets which we call clusters. Every cluster has some a special vertex which we call the \emph{cluster center}. 
In the high level, the Baswana-Sen algorithm computes $k$ levels of clustering $\mathcal{C}_0, \ldots, \mathcal{C}_{k-1}$ where each clustering $\mathcal{C}_i$ is obtained by sampling the cluster center of each cluster in $\mathcal{C}_{i-1}$ with probability $n^{-1/k}$. Each cluster $C \in \mathcal{C}_i$ has in the spanner $H$, a BFS tree of depth $i$ rooted at the cluster center spanning\footnote{The vertices of the tree are precisely the vertices of the cluster.} all the nodes of $C$. The vertices that are not incident to the sampled clusters become unclustered. For each unclustered vertex $v$, the algorithm adds one edge to each of the clusters incident to $v$ in $\mathcal{C}_{i-1}$. This randomized construction is shown to yield a spanner with  $O(kn^{1+1/k})$ edges in expectation and it can be implemented in $O(k^2)$ rounds\footnote{With some care, we believe the algorithm can also be implemented in $O(k)$ rounds}. 
Note that the only randomized step in Baswana-Sen is in picking the cluster centers of the $i^{th}$ clustering. That is, given the $n^{1-(i-1)/k}$ cluster centers $Z_{i-1}$ of the clusters in $\mathcal{C}_{i-1}$, it is required to pick $n^{-1/k}$ fraction of it, to be centers of the clusters in $\mathcal{C}_{i-1}$.

\paragraph{The brute-force deterministic solution in $O(n)$ rounds:}
 A brute-force approach to pick the new cluster centers of $\mathcal{C}_i$ is to iterate over the clusters in $\mathcal{C}_{i-1}$ one by one, checking if they satisfy some expansion criteria. Informally, the expansion is measured by the number of vertices in the $i^{th}$-neighborhood of the cluster center (i.e., number of vertices that can be covered\footnote{We say the a vertex is covered by a cluster-center if it gets into its cluster.} by the cluster center in case it proceeds to the $i^{th}$ level). 
If the expansion is large enough, the current cluster ``expands" (i.e., covers vertices up to distance $i$), and joins the $i^{th}$-level of the clustering $\mathcal{C}_{i}$. 
Since in the first level there are $O(n)$ clusters (each vertex forms a singleton cluster), this approach gives an $O(n)$-round algorithm. With some adaptations, this approach can yield an improved $O(n^{1-1/k})$ round algorithm (as in \cite{derbel2010sublinear}).

\paragraph{Our $O(n^{1/2-1/k})$-round deterministic solution:}
Inspired by the randomized construction of Baswana-Sen and the work of Derbel, Gavoille, and Peleg \cite{derbel2007deterministic}, we present a new approach for constructing spanners, based on two novel components which we discuss next.
%

\subsubsection{Key Idea (I): Grouping Baswana-Sen Clusters into Superclusters}
%
Our approach is based on adding an additional level of clustering on top of Baswana-Sen clustering. We introduce the novel notion of a supercluster -- a subset of Baswana-Sen clusters that are \emph{close} to each other in $G$. 
In every level $i\leq k/2$, we group the $O(n^{1-i/k})$ clusters of $\mathcal{C}_{i}$ into $O(\sqrt{n})$ \emph{superclusters}, each containing $O(n^{1/2-i/k})$ clusters which are also close to each other in $G$. 
Specifically, the superclusters have the following useful structure: 
cluster-centers of the same supercluster are connected in $G$ by a \emph{constant depth tree} (i.e., the weak diameter of the superclusters is $O(1)$), and the trees of different superclusters are \emph{edge-disjoint}. 
See Fig. \ref{fig:smalllarge} for an illustration.

Unlike the brute-force $O(n)$-round algorithm mentioned above, our algorithm iterates over \emph{superclusters} rather than clusters. We define the neighborhood of the supercluster to be all vertices that belong to -- or have a neighbor in-- one of the clusters of that supercluster. The expansion of the supercluster is simply the size of this neighborhood. 
The importance of having this specific structure in each supercluster is that it allows the superclusters to compute their expansion in $O(1)$ rounds\footnote{using the collection of edge disjoint $O(1)$-depth trees that connect their cluster centers.}.  
The faith of the superclusters (i.e., whether they continue on to the next level of clustering), in our algorithm, is determined by their expansion.
 If the expansion of a supercluster is sufficiently high, \emph{all} the cluster centers of that supercluster join the next level $i$ of the clustering. Otherwise, all these clusters are discarded from the clustering. As we will show in depth, the algorithm makes sure that at most $O(n^{1/2-1/k})$ superclusters pass this ``expansion test" and the remaining superclusters with low-expansion are handled using our second key tool as explained next.
%
%
%
%
%
%
\subsubsection{Key Idea (II): Better Spanners for Unbalanced Bipartite Graphs}
In our spanner construction, each supercluster with low-expansion has additional useful properties: it has $|A|=O(\sqrt{n})$ vertices and only $|B|=O(n^{1/2+1/k})$ many ``actual" neighbors\footnote{This term is informal, the actual neighbors are the vertices that are no longer clustered by the current clustering.}. We then apply Lemma \ref{bipartite} on these superclusters by computing the $(2k-1)$-spanner for each of these bipartite graphs obtained taking the vertices of the supercluster to be on one side of the bipartition, $A$, and their ``actual" neighbors on the other side $B$.
Since there are $O(n^{1/2})$ superclusters, this adds $O(n^{1/2+1/k} \cdot n^{1/2})=O(n^{1+1/k})$ edges to the spanner.
 
Finally to provide a good stretch in the spanner for all the edges in $G$ 
between vertices of the same supercluster\footnote{Vertices whose clusters belong to the same supercluster.}, we simply recurse inside each supercluster-- this can be done efficiently since the superclusters are vertex disjoint (as they contain sets of vertex disjoint clusters), and each of supercluster has $O(\sqrt{n})$ vertices.

\paragraph{Roadmap.}
The structure of the paper is as follows. We start by considering in Section \ref{sec:3spann} the simplified case of $3$-spanners (hence $k=2$) and present a deterministic construction with $O(\log n)$ rounds. 
Section \ref{sec:kspanner} considers the general case of $k > 2$. In Sec. \ref{sec:naive}, we first describe an $O(n^{1-1/k})$-round algorithm that already contains some of the ideas of the final algorithm. Then, before presenting the algorithm, we describe the two key tools that it uses.
For didactic reasons, in Sec. \ref{sec:unbalancedbipartite}, we first describe the construction of sparser spanners for unbalanced bipartite graphs. Only later \ref{sec:supercluster}, we present the new notion of superclusters. 
Finally, in Sec. \ref{sec:alg}, we show how these tools can be used to construct $(2k-1)$-spanners for graphs of low diameter. The extension for general graphs is in Appendix \ref{append:largediameter}.
\section{Preliminaries, Notation and Model}
\paragraph{Notations and Definitions.}
We consider an undirected unweighted $n$-vertex graph $G=(V,E)$ where $V$ represents the set of processors and $E$ is the set of links between them. Let $\diam(H)$ be the diameter of the subgraph $H \subseteq G$. We denote the diameter of $G$ by $D=\diam(G)$. For $u,v \in V(G)$ and a subgraph $H$, let $\dist(u, v,H)$ denote the $u-v$
distance in the subgraph $H \subseteq G$.
When $H=G$, we omit it and write $\dist(u, v)$.
Let $\Gamma(u)=\{v ~\mid~ (u,v)\in E\}$ be the set of $u$'s neighbors in $G$ and $\Gamma^+(u)=\Gamma(u) \cup \{u\}$. 
For a subset of the vertices $V' \in V$, let $\Gamma(V', G)=\bigcup_{u \in V'}\Gamma(u)$ and $\Gamma^+(V', G)= \Gamma(V', G) \cup V'$. Let $\Gamma_i(v)=\{u ~\mid~ \dist(u,v)\leq i\}$, for a subset $V'$, $\Gamma_i(V')$ is defined accordingly. 
For a subgraph $H \subseteq G$, let $E(v,H)=\{(u,v) \in E(H)\}$ be the set of edges incident to $v$ in the subgraph $H$ and let $\deg(v,H)=|E(v,H)|$ denote the degree of node $v$ in $G$. For a set of vertices $C$, let $G(C)$ be the induced graph of $G$ on $C$. We say that the a tuple $(a,b)> (c,d)$ if $a>c$ or $a=c$ but $b>d$.

\paragraph{Spanners and Clustering.}
A subgraph $H \subseteq G$ is a $(2k-1)$-\emph{spanner} if $\dist(u,v,H)\leq (2k-1)\dist(u,v,G)$ for every $u,v \in V$.  Given a graph $G$, a subgraph $H$, and an edge $e = (u, v)$ in $G$, we define the \textit{stretch} of $e$ in $H$ to be the length of the shortest path from $u$ to $v$ in $H$. If no such path exists, we say that the stretch is infinite. We say that an edge $e=(u,v)$ is \emph{taken care of} in $H$ if $\dist(u,v,H)\leq (2k-1)$.
A \emph{cluster} is a connected set $C$ of vertices of the original graph. Often, a cluster will have one of its vertices $s \in C$ be the \textit{cluster center}. The ID of the cluster is the ID of its center. Two clusters $C_1$ and $C_2$ are \emph{neighbors} if $\Gamma(C)\cap C' \neq \emptyset$. For a subset of vertices $S \subseteq V$, the diameter of the subset is simply the diameter of the induced graph $G$ on $S$. 


\paragraph{Ruling Sets.}
An \textit{$(\alpha, \beta)$-ruling set} with respect to $G$ and $V' \subseteq V$ is a subset $U \subseteq V'$ satisfying the following:
(I) All pairs $u, v \in U$ satisfy $\dist(u, v) \ge \alpha$, (II)
For all $v \in V'$, there exists a $u \in U$ such that $\dist(u, v) \le \beta$.
%

\paragraph{The Communication Model.}
We use a standard message passing model, the
\congest\ model \cite{Peleg:2000}, where the execution proceeds in synchronous rounds and in each round, each node can send a message of size $O(\log n)$ to each of its neighbors.
In this model, local computation is done for free at each node and the primary complexity measure is the number of communication rounds. Each node holds a processor with a unique and arbitrary ID of $O(\log n)$ bits.
Throughout, we assume that the nodes know a constant approximation on the number of nodes $n$, same holds also for \cite{baswana07}.

%
%
%
%
%

\section{3-Spanners in $\widetilde{O}(1)$ Rounds}
\label{sec:3spann}
The key building block of the algorithm is the construction of a linear sized $3$-spanner for a $\sqrt{n}  \times n$ bipartite graph. A similar idea already appeared in \cite{derbel2007deterministic}, but using $O(n)$-bit messages.

\paragraph{The core construction: 3-spanners for unbalanced bipartite graphs}

\begin{lemma}\label{bipartite3}
Let $G = (A, B, E)$ be a (possibly weighted) bipartite graph where $|A| = O(\sqrt{n})$, $|B| = O(n)$, and each vertex knows whether it is in $A$ or in $B$. Then one can construct (in the \congest\ model) a $3$-spanner $H$ with $O(n)$ edges within \emph{two} rounds.
\end{lemma}
Algorithm $\BipartiteThreeSpanner$ first forms $|A|$ vertex-disjoint star clusters (clusters of radius 1), each centered at a vertex of $A$.  
To do that, every vertex $v_b \in B$ picks one of its neighbors $v_a \in \Gamma(v_b) \cap A$ to be its cluster center and sends the ID of its chosen neighbor $v_a$ to all of its neighbors. We write $c(v_b)=v_a$ to denote that the cluster center of $v_b$ is $v_a$. 

All edges 
$(v_b, c(v_b))$ are added to the spanner $H$. At this point, the graph contains $O(\sqrt{n})$ clusters centered at the vertices in $A$. We say that two stars $S_1$ and $S_2$ are \emph{neighbors} if the \emph{center} of $S_1$ has a neighbor in $S_2$, or vice-versa. Note that because the graph is bipartite, this is the only possible connection between clusters.
Then, for each vertex $u_a$ in $A$, and each neighboring star-cluster $u'_a$, the vertex $u_a$ adds to the spanner $H$ one edge to one of its neighbors in the cluster of $u'_a$. For a complete description of the algorithm see Alg. \ref{alg:bipartite3}.  For illustration see Fig. \ref{fig:lowerboundtwof}.
\begin{algorithm}
\caption{\small{$\BipartiteThreeSpanner(G = (A \cup B, E))$ for $|A| = O(\sqrt{n})$ and $|B| = O(n)$}}
\begin{algorithmic}[1]
\algorithmsize
\State $H\gets \emptyset$
\State Each vertex $v_b \in B$ selects an arbitrary neighboring vertex $v_a \in A$, assigns $c(v_b)=v_a$ and send $c(v_b)$ to all its neighbors. It adds the edge $(v_b, c(v_b))$ to $H$.
\State Each vertex $u_a \in A$ does the following (in parallel):	
	\For{each ID $v_a$ received}
		\indent \State Pick a single neighbor $v_b$ satisfying $c(v_b) = v_a$. Add the edge $(u_a, v_b)$ to $H$.
	\EndFor
\end{algorithmic}
\label{alg:bipartite3}
\end{algorithm}
\begin{figure}[h!]
\begin{center}
\includegraphics[scale=0.3]{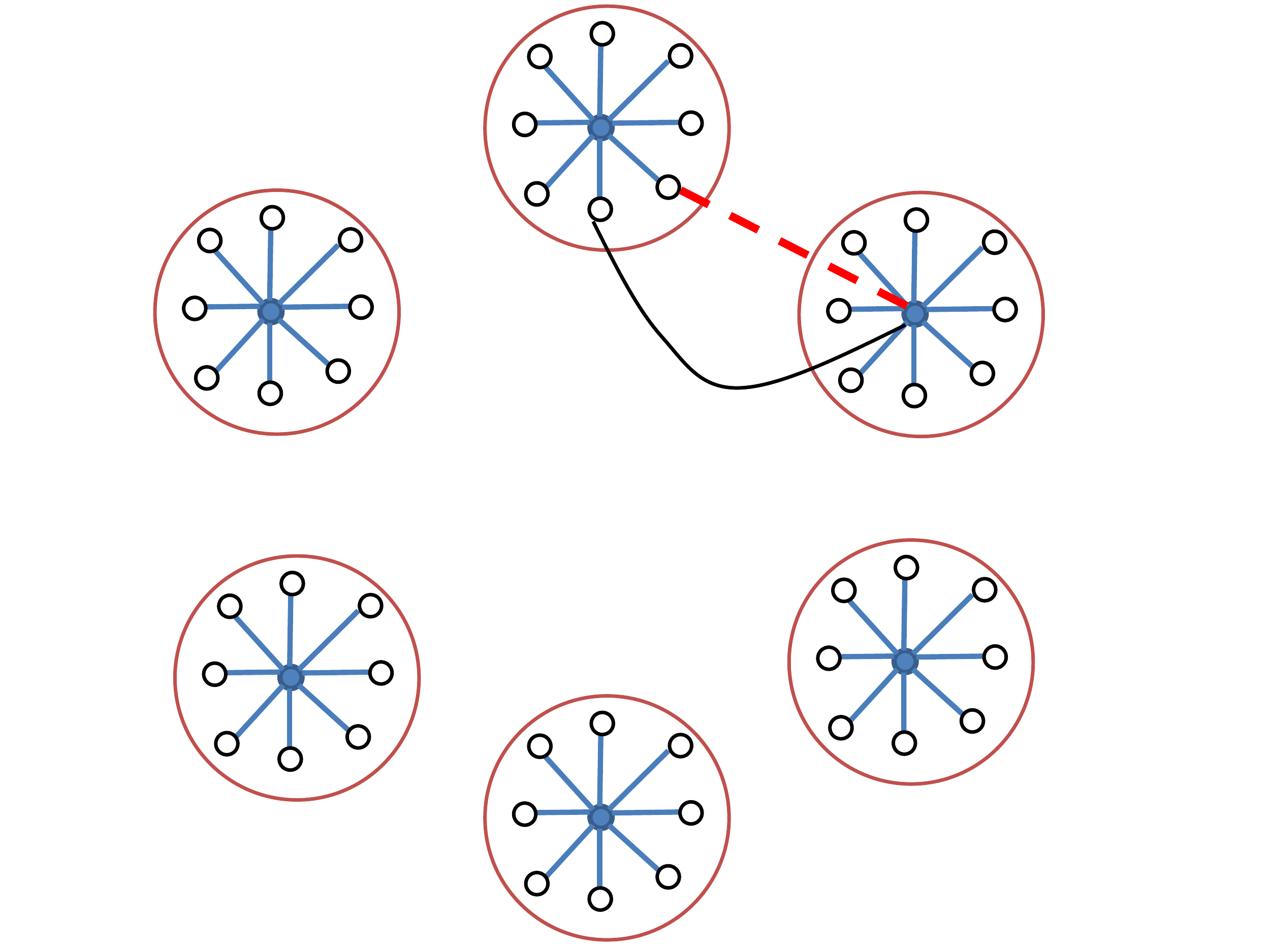}
\caption{
\label{fig:lowerboundtwof} Illustration of the star decomposition and the spanner. The dashed edge that is omitted from the spanner has a path of length three in the spanner.}
\end{center}
\end{figure}
To adapt the algorithm for the weighted case, we simply let each $v_b \in B$ pick its closest neighbor in $A$. In addition, each vertex $u_a$ connects to its closest neighbor in each star of $u'_a$. 
It is easy to see that the algorithm takes 2 rounds. 
\begin{lemma}
\label{alg:threespaneasy}
The output $H$ of  Alg. $\BipartiteThreeSpanner$ is a $3$-spanner with $O(n)$ edges.
\end{lemma}
\begin{proof}
First, we show that $H$ is a 3-spanner. 
To analyze the stretch, consider an edge $(x,y)\in G\setminus H$ where $x \in A$ and $y \in B$. Let $z=c(y)$. Since $(x,y) \notin H$, we have that $\dist(z,y,G)\leq \dist(x,y,G)$. Let $w \in B$ be the neighbor of $x$ in the cluster of $z$ such that the edge $(x,w)$ was added to $H$. Since $(x,y) \notin H$, such an edge $(z,w)$ exists and in addition, $\dist(x,w,G)\leq \dist(x,y,G)$. Overall, we have: 
\begin{eqnarray*}
\dist(x,y,H)&\leq& \dist(x,w,H)+\dist(w, z, H) + \dist(z,y,H)
\\&=&\dist(x,w,G)+\dist(w, z, G)+\dist(z,y,G)
\\&\leq & 
\dist(x,y,G)+\dist(w, x, G)+ \dist(x, y, G)
\\&\leq& 
\dist(x,y,G)+\dist(y, x, G)+\dist(x, y, G)=3\dist(x,y,G)~.
\end{eqnarray*}
We next bound the number of edges in $H$. In the first round (line 2 of Alg. \ref{alg:bipartite3}), each element of $B$ adds at most one edge to $H$, for a total of $|B|$ edges. In the second round, each element of $A$ added at most $|A|$ elements to $|H|$. Overall, $H$ has $|B| + |A|^2 = O(n)$ edges.
\end{proof}

\paragraph{Constructing $3$-spanners for general graphs in $O(\log n)$ rounds:}
Let $V_h=\{ v \in V ~\mid~ \deg(v,G)\geq \sqrt{n}\}$ be the set of \emph{high} degree vertices in $G$ and let $V_\ell=V \setminus V_h$ be the remaining \emph{low}-degree vertices. First, the algorithm adds to the spanner $H$, all the edges of the low-degree vertices $V_{\ell}$. Then, it proceeds by partitioning (in a way that will be described later) the \emph{high-degree} vertices $V_h$ into $t=O(\sqrt{n})$ balanced sets $V_1, \ldots, V_{t}$. This partition gives rise to $t$ bipartite $\sqrt{n} \times n$ graphs $B_i$ obtained by taking $V_i$ to be on one side of the partition and $V\setminus V_i$ on the other side. We describe the partitioning procedure in Lemma \ref{3Partition}. We can then apply Algorithm $\BipartiteThreeSpanner$ to construct $3$-spanners for all these subgraphs in parallel. 
Finally, we simply add to $H$, all the internal edges $V_i \times V_i$ for every $i$, again adding total of $t \cdot O(n)$ edges. 
\newpage
\begin{mdframed}[hidealllines=false,backgroundcolor=gray!30]
\centering  \textbf{\large  Algorithm $\ImprovedThreeSpanner$} \\
\begin{flushleft}
\vspace{-5pt}
\begin{itemize}
\item{\textbf{(S0): Handling Low-Degree Vertices:}} Add to $H$ all edges in $(V_\ell \times V) \cap E(G)$. 
\item{\textbf{(SI): Balanced Partitioning of High-Degree Vertices $V_h$:}}
Partition the high-degree vertices of $V$ into $\Theta(\sqrt{n})$ sets $V_1, \ldots, V_t$ each with $O(\sqrt{n})$ vertices.
\item{\textbf{(SII): Taking care of edges $V_i \times (V \setminus V_i)$, for every $i \in \{1,\ldots, t\}$}:}
\begin{itemize}
\item Define $B_i=(V_i,V\setminus V_i)$ for every $i \in \{1,\ldots, t\}.$
\item Apply Algorithm $\BipartiteThreeSpanner$ on each $B_i$ graph in parallel to construct a $3$-spanner $H_i \subseteq B_i$ for every $i$ and these $H_i$ subgraphs to $H$. 
\end{itemize}
\item{\textbf{(SIII): Taking care of edges $V_i \times V_i$, $i \in \{1,\ldots, t\}$}:}
Add to $H$ all edges in $V_i \times V_i$ for every $i \in \{1,\ldots, t\}$. 
\end{itemize}
\end{flushleft}
\end{mdframed}
Note that eventhough the bipartite graphs $B_i$ are not vertex disjoint, each edge belongs to at most two such graphs, and hence we can construct the $3$-spanners for all $B_i$ in parallel. It is also easy to see that the final spanner has $O(n^{3/2})$ edges.

The only missing piece at that point concerns the computation of partitioning $V_h$.

\paragraph{Balanced partitioning of $V_h$ in $O(\log n)$ rounds:}
The partition procedure starts by computing $(4,O(\log n))$-ruling set $R \subseteq V$ for the high-degree vertices $V_h$. We will use the following lemma that uses standard technique for constructing $(t, t\log n)$-ruling sets. 
\begin{lemma}\label{3Partition}
Given a graph $G = (V, E)$ and a subset $V_h \in V$ of the vertices,
one can compute in $O(\log n)$ rounds in the \congest\ model, a $(4,O(\log n))$-ruling set $U \subseteq V_{h}$ with respect to $G$ and the high-degree vertices $V_h$. 
\end{lemma}
\def\APPENDLEMMARULING{
Our proof is similar to the standard construction of a $(4,O(\log n))$-ruling set 
\cite{panconesi1992improved}.
We begin by partitioning the vertices into two sets $V_0$ and $V_1$, based on the last digit of their IDs. Then, we recursively apply the algorithm to $V_0$ and $V_1$ to create ruling sets $V_0'$ for $V_0$ and $V_1'$ for $V_1$. Next, we remove all elements of $V_1'$ which are at distance less than $4$ from $V_0'$. Note that this increases the maximum distance from a vertex to $V'$ by at most $4$. Since the IDs are of length $O(\log n)$, the depth of the recursion is $O(\log n)$. At every step, the maximum distance between two vertices in the ruling set increases by $O(1)$, so in the final output the maximum distance between two vertices in the ruling set will be $O(\log n)$.
}
We now view each of the vertices $r \in R$ as a center of a cluster of diameter $O(\log n)$: let each high-degree vertex join the cluster of the vertex closest to it in $R$, breaking ties based on IDs. Since every vertex in $V_h$ is at distance $O(\log n)$ from $R$, all the vertices $V_h$ will be clustered within $O(\log n)$ rounds.
Each vertex $r$ in $R$ can then partition the vertices of its cluster into subsets of size $\lfloor \sqrt{n} \rfloor$, and an additional leftover subset of size at most $\sqrt{n}$ (this can be done using balanced partitioning lemma, Lemma \ref{lem:balancedpartitioning}). We now claim that this partition is balanced. Clearly, all sets are of size $O(\sqrt{n})$, so we just show that there $O(\sqrt{n})$ subsets.
Since every $r \in R$ is high-degree and since every two vertices in $R$ are at distance at least $4$, we have that $|R| = O(\sqrt{n})$. For each $r \in R$, there is at most one subset of size less than $\lfloor \sqrt{n} \rfloor$. Therefore, there are $O(\sqrt{n})$ subsets of size less than $\lfloor \sqrt{n} \rfloor$. All other subsets are of size $\lfloor \sqrt{n} \rfloor$. However, there can be at most $O(\sqrt{n})$ disjoint subsets of size $\lfloor \sqrt{n} \rfloor$, hence there are $O(\sqrt{n})$ subsets in total, as desired.

We conclude by showing:
\begin{lemma}[3-Spanner Given Partition]\label{3givenpartition}
Given a (possibly weighted) $n$-vertex graph $G = (V, E)$ with a vertex-partition $V_1, V_2, \ldots, V_{t}$ such that $|V_i| = O(\sqrt{n})$ and $t= O(\sqrt{n})$, one can construct a 3-spanner $H$ of size $O(n^{3/2})$ in 2 rounds in the \congest\ model.
\end{lemma}
\begin{proof}
For each $i$, we construct a spanner $H_i$ for the bipartite graph $B_i=(V_i \cup (V \setminus V_i))$. Since each edge is in at most two of the bipartite graphs, all these $t=O(\sqrt{n})$ spanners $H_1, \ldots, H_t$ can be constructed in parallel, while slowing down by a factor of at most $2$.
Let $H =\bigcup_{i} H_i$, we have $|H| = O(n^{3/2})$.
In addition, we add to $H$ all the internal edges inside each $V_i$. This adds a total of at most $\sum_{i=0}^{\ell}|V_i|^2$ edges to $H$. Since $|V_i|^2 = O(n)$, and $\ell = O(\sqrt{n})$, the total number of edges added is $\sum_{i=0}^{\ell}|V_i|^2 = O(n^{3/2})$, and $|E(H)| = O(n^{3/2})$.

We now prove that $H$ is a 3-spanner. Consider some edge $(u, v)$. If both $u$ and $v$ are in the same $V_i$, then $(u, v) \in H$. If $u \in V_i$ and $v \in V_j$, then, consider the spanner $H_i$ computed for the graph $G_i=(V_i \cup (V \setminus V_i))$. Since $(u, v)$ is an edge in $G_i$, there must exist a path of length at most $3$ from $u$ to $v$ in $H_i$, and therefore in $H$. The lemma follows.
\end{proof}
Finally, if the vertices IDs are bounded, two rounds are sufficient to construct the spanner. 
\begin{theorem}[Small IDs]\label{IDs}
Given a graph $G = (V, E)$ where the IDs of the vertices have $\log(n) + O(1)$ bits, one can construct a 3-spanner $H$ of $G$ with $|H| = O(n^{3/2})$ edges in two rounds.
\end{theorem} 
\begin{proof}
We can use the second half of the IDs of the vertices to partition $V$ into subsets $V_1, V_2, \ldots, V_{\ell}$. There are $2^{(\log(n) + O(1))/2} = O(\sqrt{n})$ possibilities for the second half of the ID. Hence, $\ell = O(\sqrt{n})$.
Similarly, there are $O(\sqrt{n})$ possibilities for the first half of the ID, so $|V_i| = O(\sqrt{n})$ for all $i$. We now apply Lemma \ref{3givenpartition} on $G$ with the sets $V_1, V_2, \ldots, V_{\ell}$. 
\end{proof}

\vspace{-10pt}
\section{$(2k-1)$ Spanners}\label{sec:kspanner}
\vspace{-8pt}
\paragraph{The structure of Baswana-Sen clustering.}
At the heart of the algorithm is a construction of $(k-1)$-levels of clustering $\mathcal{C}_0, \ldots, \mathcal{C}_{k-1}$. The initial clustering $\mathcal{C}_0=\{\{v\}, v \in V\}$ simply contains $n$ singleton clusters. For every $i$, each cluster $C \in \mathcal{C}_i$ has a cluster center $z$ and we denote by $Z_i$ the collection 
of cluster centers. We define $V_i=\bigcup_{z \in Z_i}\Gamma_i(z)$. A vertex $v$ is $i$-\emph{clustered} if $v \in V_i$, otherwise it is $i$-unclustered.  Hence $V_i$ is the set of clustered vertices appearing in the clusters of $\mathcal{C}_i$. 
The algorithm consists of $k-1$ steps where at the end of step $i \in \{1, \ldots, k-1\}$, we have
an $i^{th}$-level clustering $\mathcal{C}_{i}=\{C_1, \ldots, C_{\ell}\}$ and a partial spanner $H_i$ that satisfies the following: 
(P1) The clustering $\mathcal{C}_i$ contains $\ell=O(n^{1-i/k})$ clusters.
(P2) For each cluster $C_j \in \mathcal{C}_i$ with a center $z_j$, the subgraph $H_i$ contains a BFS tree $T_i(C)$ of depth at most $i$ that spans all the vertices of $C$ (i.e., the vertices of $T_i(C)$ are precisely $C$)  and (P3) For every $u \in V_{i-1}\setminus V_i$, and every $v \in \Gamma(u)$,  $\dist(u,v,H_i) \leq 2k-1$. 
\begin{mdframed}[hidealllines=false,backgroundcolor=gray!30]
\centering  \textbf{\large  High-Level Description of Phase $i$ in Baswana-Sen Algorithm} \\
\begin{flushleft}
\vspace{-5pt}
\begin{itemize}
\item{(SI)}
\textbf{Selecting $O(n^{1-i/k})$ cluster centers $Z_i \subseteq Z_{i-1}$}. In the randomized algorithm, this is done by sampling each center in $Z_{i-1}$ independently with probability $n^{-1/k}$. 
We then have $V_i=\Gamma^+_i(Z_i)$ are the $i$-clustered vertices.
\item{(SII)}
\textbf{Taking care of unclustered vertices $V_{i-1}\setminus V_{i}$}. That is, taking care of the vertices that stopped being clustered at that point. 
\item{(SIII)}
\textbf{Forming the clusters of $\mathcal{C}_{i}$ around $Z_i$}. This is done by letting each $u \in V_i$ join the cluster of its \emph{closest} center in $Z_i$ breaking tie based on ID's. The latter can be implemented in $O(i)$ rounds of constructing BFS trees of depth $i$ from all centers $Z_i$ while breaking ties appropriately.
\end{itemize}
\end{flushleft}
\end{mdframed}
At the final phase of Baswana-Sen, there are $O(n^{1/k})$ clusters in $\mathcal{C}_{k-1}$ and at that point, each vertex $v \in V$ adds one edge to each of its neighboring clusters in $\mathcal{C}_{k-1}$. 

Note that the only step that uses randomness in this algorithm is sub-step (SI), and the other two sub-steps (SII-SIII) and the final phase are completely deterministic. 
Our challenge is to implement sub-step (SI) deterministically in a way that in sub-step (SII) we do not add too many edges to the spanner. The algorithms presented from now on, will simulate the $i^{th}$ phase of Baswana-Sen only without using randomness. Sub-step (SIII) and the final phase will be implemented exactly as in Baswana-Sen. 
\vspace{-10pt} 
\subsection{Take (I): $O(n^{1-1/k})$-Round Algorithm $\NaiveSpanner$}\label{sec:naive}
It is easy to see that $0^{th}$-level clustering containing $n$ singleton clusters satisfies properties (P1-P3). To simulate the $i^{th}$ phase of Baswana-Sen algorithm, we employ $O(i \cdot n^{1-i/k})$ deterministic rounds:
Initially, we unmark all the vertices and over time, some of the vertices will get marked (i.e., indicating that they are $i$-clustered). The procedure consists of $O(n^{1-i/k})$ steps where at each step, we look at the remaining set $Z'_{i-1}$ of cluster centers in $Z_{i-1}$ that have not yet been added to $Z_i$. 
Let $U$ be the current set of unmarked vertices and let $\mathcal{C}'_{i-1} \subseteq \mathcal{C}_{i-1}$ be the corresponding clusters of $Z'_{i-1}$. 
For each cluster $C \in \mathcal{C}_{i-1}$, define its unmarked neighborhood by $\Gamma^U(C)=\bigcup_{u \in C}\Gamma(u)\cap U$ and its current unmarked-degree by $\deg^U(C)=|\Gamma^U(C)|$. We say that cluster $C$ is a \emph{local-maxima} in its unmarked neighborhood if it has the maximum tuple (lexicographically) $(\deg^U(C), ID(C))$ among all other clusters $C'$ that have mutual unmarked neighbors (i.e., $\Gamma^U(C)\cap \Gamma^U(C')\neq \emptyset$). 
\begin{mdframed}[hidealllines=false,backgroundcolor=gray!30]
\centering  \textbf{\large  Phase $i$ of Algorithm $\NaiveSpanner$} \\
\begin{flushleft}
\vspace{-5pt}
\textbf{(SI): Defining the centers $Z_i$.}\\
Set $Z'_{i-1}\gets Z_{i-1}$, $U=V$ and for $O(n^{1-i/k})$ steps do the following: 
\begin{itemize}
\item Every center $z \in Z'_{i-1}$ of cluster $C$ computes $\deg^U(C)$.
\item Every center $z \in Z'_{i-1}$ whose cluster $C$ has the maximum tuple $(\deg^U(C), ID(C))$ in its unmarked neighborhood, $\deg^U(C)$, joins $Z_{i}$ only if $\deg^U(C)\geq n^{i/k}$.
\item Remove from $Z'_i$ the centers $z \in C$ that join $Z_i$ and mark $\Gamma^U(C)$.
\end{itemize}
\textbf{(SII): Taking care of unclustered vertices.}
\begin{itemize}
\item
Let $\mathcal{C}'_{i-1}$ be the clusters whose centers did \emph{not} join $Z_{i}$. 
\item
Every remaining unmarked vertex $u$, adds one edge per neighboring cluster in $\mathcal{C}'_{i-1}$.
\end{itemize}
\textbf{(SIII): Forming the $\mathcal{C}_i$ clusters centered at $Z_i$.} As in Baswana-Sen.
\end{flushleft}
\end{mdframed}
\def\APPENDFullDescriptionSimple{
\subsection{$O(n^{1-1/k})$-Round Algorithm $\NaiveSpanner$: Missing Details}
We first let each cluster center $z$ compute the degree $\deg^U(C)$ of its cluster $C$. To do that, every clustered vertex $v \in C$ for every $C \in \mathcal{C}'_{i-1}$ sends its cluster ID to its neighbors and every unmarked vertex $u$ sends an ACK to one neighbor in each cluster. 
By that, we avoid the double counting due to nodes that have many neighbors in the same cluster. Then, 
the clustered vertices upcast the number of ACKs they received on the depth-$i$ BFS tree rooted at their center. 
Hence, within $O(i)$ rounds, each node $v \in C$ knows $\deg^U(C)$ for every $C \in \mathcal{C}'_{i-1}$. 
Next, we want to elect cluster centers whose unmarked degree (of their clusters) is local-maxima in their $i^{th}$-neighborhood. Only these cluster centers will be given the chance to join the $i^{th}$-clustering. To do that, each clustered node $v \in C$ for every $C \in \mathcal{C}'_{i-1}$ sends to its unmarked neighbors the tuple $(\deg^U(C),ID(C))$. Every unmarked node sends ACK only to the unique neighbor $v$ from which it got the tuple of maximum value. These ACKs are then upcasted on the communication tree within each cluster. Cluster centers $c \in C$ that got an ACK from all their unmarked neighbors $\Gamma^U(C)$ will now join the next level of the clustering only if $\deg^U(C)\geq n^{i/k}$. If this condition holds, we add these cluster centers to $Z_i$ and mark their $i^{th}$-neighborhood. This process continues for $O(n^{1-i/k})$ steps, it is easy to see that after at that point all remaining clusters $C$ have small $\deg^U(C) \leq n^{i/k}$ and that $Z_i$ contains only $O(n^{1-i/k})$ cluster centers\footnote{Since the unmarked neighborhood $\Gamma^U(C), \Gamma^U(C')$ of two clusters that join the $i^{th}$-level clustering are \emph{disjoint}, there $n^{1 - (i-1)/k}$ such clusters in the clustering $\mathcal{C}_{i}$.}. This completes the description of sub-step (SI). 
Next, step (SII) is implemented as follows: every unmarked vertex $u$ add one edge to each of its neighboring remaining clusters in $\mathcal{C}_{i-1}$. Since we have $O(n^{1-(i-1)/k})$ such clusters and each of them has $O(n^{i/k})$ unmarked neighbors, this adds $O(n^{1+1/k})$ edges to the spanner\footnote{Note that this algorithm takes care of all edges $(V_{i-1}\setminus V_i)\times (V\setminus V_i)$. In Baswana-Sen, the algorithm takes care of all edges $(V_{i-1}\setminus V_i)\times V$, however, in the analysis section, we show that our modification is still sufficient to have a bound stretch.}. 
Finally, sub-step (SIII) of forming the clusters centered at $Z_i$ is exactly as in Baswana-Sen.
}
\textbf{Sketch of the Analysis:} The key part to notice is that by picking the local-maxima clusters, we have that for any two cluster-centers $z_1 \in C_1, z_2 \in C_2$ that join $Z_i$, their unmarked neighborhoods $\Gamma^U(C_1),\Gamma^U(C_2)$ are vertex disjoint, hence $Z_i$ contains $O(n^{1-i/k})$ centers; in addition, after $O(n^{1-i/k})$ steps, the clusters of all remaining centers have $O(n^{i/k})$ unmarked neighbors. Hence, at step (SII), total of $O(n^{1-(i-1)/k})\cdot O(n^{i/k})=O(n^{1+1/k})$ edges are added to the spanner.  
Turning to runtime, we claim that each of the $O(n^{1-i/k})$ steps can be implemented in $O(i)$ rounds. Since each cluster $C \in \mathcal{C}_{i-1}$ is connected i $G$ by a depth-$i$ tree, and since trees of different clusters are vertex-disjoint, computing the unmarked degree $\deg^U(C)$ of each cluster $C$ can be done in $O(i)$ rounds; To avoid the double of counting of unmarked vertices that have many neighbors at the same cluster, each unmarked vertex respond to only one its neighbors in each cluster. Similarly, also selecting the local maxima clusters can be done in $O(i)$ rounds.
We note that the time complexity of the algorithm is $O(n^{1-1/k})$, as opposed to $O(n)$, since after $O(n^{1-1/k})$ iterations of finding clusters of locally maximal unmarked degree, all remaining clusters will have low unmarked degree, and can be dealt with in parallel in $O(1)$ rounds, by adding an edge to every unmarked neighbor. Towards speeding up this algorithm, we now introduce our key technical tools.

\dnsparagraph{A remark regarding step (SII):} Let $V'_{i}=V_{i-1} \setminus V_{i}$ be the set of newly unclustered vertices. In Baswana-Sen algorithm, step (SII) takes care of all the edges in $V'_{i}\times V$. That is, the edges added to the spanner $H$ at that stage provide that $\dist(u,v,H)\leq 2i-1$ for every $(u,v) \in (V'_{i}\in V)\cap E$. Most of the algorithms we present in this paper, have a weaker but sufficient guarantee when implementing step (SII). In particular, we only add edges between the remaining unmarked vertices and the remaining clusters whose centers did not join $Z_i$. 
We now show why it is sufficient. Consider an edge $(u,v) \in E$. Let $i_u$ be the largest level of the clustering such that $u$ is $i_u$-clustered, define the same for $v$. Without loss of generality, assume that $i_v \leq i_u$.
\textbf{Case (1): $i_u=k-1$:} Let $C$ be the cluster of $u$ in $\mathcal{C}_{k-1}$. Since in the last step, $v$ adds one edge to $\Gamma(v)\cap C$, the claim holds. 
\textbf{Case (2): $i_u\leq k-2$:} Consider phase $(i_u+1)$ where the clustering $\mathcal{C}_{i_u+1}$ is constructed given $\mathcal{C}_{i_u}$. By definition, in step (SII) of phase $(i_u+1)$ we have that the vertex $v$ is unmarked and the vertex $u$ belongs to a remaining cluster $C \in \mathcal{C}_{i_u}$. Since every unmarked vertex adds one edge to each remaining cluster, we have that $v$ added one edge to $C \cap \Gamma(v)$. The claim follows.

\vspace{-10pt}
\subsection{Key Tool (I): Sparser Spanner for Unbalanced Bipartite Graphs}
\label{sec:unbalancedbipartite}In this section, we consider Lemma \ref{bipartite}.
Similarly to the construction of $3$-spanners in Section \ref{sec:3spann}, a key ingredient in our algorithm is the construction of \emph{sparser} spanners for unbalanced $A \times B$ bipartite graphs for $|A|\leq |B|$. The algorithm of \cite{derbel2007deterministic} constructs a $(2k-1)$ spanners for these bipartite graphs with $O(|A||B|^{2/k})$ edges in the \local\ model, using large messages. Our algorithm is slower than that of \cite{derbel2007deterministic}, but has the benefit of obtaining a sparser $(2k-1)$-spanner with only $O(k|A|^{1+2/k}+|B|)$ edges and while using $O(\log n)$-bit messages.

The high-level strategy of Alg. $\BipartiteSpanner$ is to first compute $|A|$ star clusters (clusters of radius 1) by letting each vertex of $B$ join an arbitrary neighbor in $A$. Hence, after one step of clustering, we have $|A|$ clusters rather than $O(n^{1-1/k})$ clusters is in Baswana-Sen. 
We then consider \textit{star graph} $G_S$ obtained contracting each star into a vertex, and essentially apply Alg. $\NaiveSpanner$ on the star-graph $G_S$ to construct a $(k-1)$-spanner $H_S\subseteq G_S$ with $O(|A|^{1+2/k})$ edges within $O(k|A|^{1 - 2/k})$ rounds. To get a $(2k-1)$ spanner $H \subseteq G$ from $H_S$, for every star-edge $(S_i,S_j) \in H_S$, add a single edge in $(S_1\times S_2)\cap E$ to $H$. Finally, adding the star edges to the spanner, gives a total of $O(|A|^{1+2/k}+|B|)$ edges. Simulating Alg. $\NaiveSpanner$ on the star-graph in the \congest\ model requires some effort. We now provide the description of Alg. $\BipartiteSpanner$ and its analysis.

\subsection*{Full Description of Algorithm $\BipartiteSpanner$}\label{sec:appbipartite}
Given a bipartite graph $G = (A \cup B, E)$ and even integer $k\geq 4$, 
the goal of this section is to construct in $O(|A|^{1 - 2/k})$ rounds, a $(2k-1)$-spanner $H \subseteq G$ with $O(|A|^{1 + 2/k} + |B|)$ edges.

We note that very often in our constructions, we apply this algorithm on graphs with $|A| = O(n^{1/2})$ and $|B| = O(n^{1/2 + 1/k})$. In this case, we get a $(2k-1)$-spanner containing $O(n^{1/2 + 1/k})$ edges within $O(kn^{1/2 - 1/k})$ rounds.

In the first step of the algorithm, each vertex $b$ in $B$ picks a neighboring vertex $a_i$ in $A$, and joins the star centered at $a_i$. We call $a_i$ the \emph{leader} of its star $S_i$. Now, our graph is partitioned into $|A|$ vertex-disjoint stars. The idea, at the high level, is to implement the $\NaiveSpanner$ on the star-graph (the graph obtained from contracting each of the stars into a vertex), with parameter $k' = k/2$ to obtain a spanner $H'$ on the star-graph. Then, this spanner $H'$, along with the edges within stars, are combined to obtain the final spanner $H$. Hence, we should output a spanner with $O(|A|^{1 + 1/k'}) = O(|A|^{1 - 2/k})$ edges within $O(|A|^{1 - 1/k'}) = O(|A|^{1 - 2/k})$ rounds. Implementing $\NaiveSpanner$ on the star-graph instead of on the given graph involves some subtleties, which we elaborate on in this section.


\paragraph{Clustering in the Star-Graph.}
Let $S_i \subseteq V$ be the set of vertices in the star centered at $a_i$ and $\mathcal{S}=\{S_i ~\mid~ a_i \in A\}$ be the set of all stars. Many times we call $a_i \in S_i$ the \emph{leader} of the star.
For every vertex $v$ and a collection of stars $\mathcal{S}' \subseteq \mathcal{S}$, let $\Gamma_{\mathcal{S}'}(v)=\{S_i \in \mathcal{S}' ~\mid~ \Gamma^+(v)\cap S_i \neq \emptyset\}$ be the neighboring stars of vertex $v$ in $\mathcal{S}'$. Throughout, a cluster is defined by a collection of stars, that is $C \subseteq \mathcal{S}$. A clustering $\mathcal{C}_i$ in the $i^{th}$-level is a collection of $O(|A|^{1-i/k'})$ star-disjoint clusters (where $k' = k/2$), each of radius $i$ in the star-graph. 
We say that a star $s_j \in \mathcal{S}$ is $i$-clustered if it belongs to one of the clusters in $\mathcal{C}_i$. Otherwise, it is an $i$-unclustered star. 

For a cluster $C$, let $V(C)=\bigcup_{S \in C}S$ be the vertices in the stars contained in $C$. Even though our cluster is a collection of stars, we still morally think of it as a collection of vertices belonging to these stars. Also, the center of the cluster is a real vertex $a \in A$ and not a star. 
We say that a cluster $C$ is a \emph{neighbor} of star $S_j$ if 
$$V(C) \cap \bigcup_{u \in S_j}\Gamma^+(u) \neq \emptyset.$$
For a collection of stars $\mathcal{S}' \subseteq \mathcal{S}$ and a cluster $C$, define the neighboring star of $C$ in $\mathcal{S}'$ by
\begin{equation}\label{eq:clustern}
\Gamma_{\mathcal{S}'}(C)=\bigcup_{v \in V(C)}\Gamma_{\mathcal{S}'}(v)~.
\end{equation}
For a clustering $\mathcal{C}_{i-1}$ and a star $S_j \in \mathcal{S}$, define the neighboring clusters in $\mathcal{C}_{i-1}$ of star $S_j$ by
\begin{equation}\label{eq:neihinclustering}
\Gamma(S_j, \mathcal{C}_{i-1})=\{ C \in \mathcal{C}_{i-1} ~\mid~ S_j \in \Gamma_{\mathcal{S}}(C)\}~.
\end{equation}
For a collection of unmarked stars $\mathcal{U} \subseteq \mathcal{S}$, define 
the unmarked star-neighborhood of a cluster $C$ by 
$$\Gamma^{\mathcal{U}}(C)=\bigcup_{v \in V(C)}\Gamma_{\mathcal{U}}(v) \mbox{~~and~~} \deg^{\mathcal{U}}(C)=|\Gamma^{\mathcal{U}}(C)|.$$

A cluster $C$ is a \emph{local-maxima} in its unmarked star-neighborhood, if for every other cluster $C'$ with $\Gamma^{\mathcal{U}}(C)\cap \Gamma^{\mathcal{U}}(C')\neq \emptyset$, it holds that
$$ (\deg^{\mathcal{U}}(C),ID(C)) > (\deg^{\mathcal{U}}(C'),ID(C')).$$
In our algorithm, the clusters $C$ compute a $2$-approximation $\appdeg^{\mathcal{U}}(C)$ of their unmarked star-degree $\deg^{\mathcal{U}}(C)$. Hence, instead of looking for clusters which are local-maxima, we will look for clusters which are approximate local maxima. That is, clusters $C$ satisfying 
$$ (\appdeg^{\mathcal{U}}(C),ID(C)) > (\appdeg^{\mathcal{U}}(C'),ID(C'))$$
for every other $C'$ with $\Gamma^{\mathcal{U}}(C)\cap \Gamma^{\mathcal{U}}(C')\neq \emptyset$.

As in Alg. $\NaiveSpanner$, the centers of the clusters join the next level of the clustering only if (1) their clusters are (approximate) local-maxima in their unmarked star-neighborhood; and (2) their (approximate) unmarked star-degree is large enough. These two conditions guarantee that only a bounded number of cluster centers joins the next level (as each of these clusters has a large set of neighboring stars, and all these sets are disjoint).

\paragraph{The $0^{th}$-level clustering $\mathcal{C}_0$.}
Recall that the collection of (vertex-disjoint) stars $\mathcal{S}$ is defined by letting each vertex in $b$ join star $S_i$ of one of its neighbors $a_i$ in $A$. The $0^{th}$-level of clustering $\mathcal{C}_0=\{\{S_1\}, \ldots, \{S_{\ell}\}\}$ consists of all singleton stars. 

We now explain how to compute the $i^{th}$ level of the clustering $\mathcal{C}_{i}$ given the $(i-1)^{th}$ level $\mathcal{C}_{i-1}$ within $O(k \cdot |A|^{1-i/k'})$ rounds. Note the star-graph is in our mind mainly for the purpose of intuition, when implementing it, the centers $Z_i \subseteq A$ are still vertices in the original graph.

\subsubsection*{The $i^{th}$ Phase of Algorithm $\BipartiteSpanner$}

\paragraph{(S1): Selecting the cluster centers $Z_i$ of the $i^{th}$-clustering.}
We employ the following iterative process for $|A|^{1 - i/k'}$ iterations. Initially, unmark all the stars of $\mathcal{S}$ by letting $\mathcal{U}=\mathcal{S}$. The next steps are similar to those of Alg. $\NaiveSpanner$ with the only exception that we unmark stars rather than vertices.

Let $\mathcal{U}' \subseteq \mathcal{S}$ be the set of current unmarked stars and let $Z'_{i-1} \subseteq Z_{i-1}$ be the set of cluster centers that have not yet joined $Z_i$. 
As in Alg. $\NaiveSpanner$, each cluster center $z_j \in Z'_{i-1}$ computes (an approximation of) the unmarked star-degree $\deg^{\mathcal{U}}(C_j)$ of its cluster $C_j$. 

\begin{lemma}\label{unmarkeddeg}
Every cluster center $z \in C$ can compute deterministically
a $2$-approximation for $\deg^{\mathcal{U}}(C)$ within $O(\min(|A|^{1-1/k'}, |A|^{1 - (i-1)/k'}))$ rounds. In addition, once it is computed in the first iteration of the phase, it can be maintained in each future iteration of that phase, using $O(k)$ rounds.
\end{lemma}

\begin{proof}
We distinguish between the first phase and phase $i$ for $i\geq 2$. In the first phase, the cluster centers $z \in C$ compute only a $2$-\emph{approximation} $\appdeg^{\mathcal{U}}(C)$ for the unmarked star-degree of their cluster. In phase $i\geq 2$, the cluster centers $z \in C$ compute the \emph{exact} value $\deg^{\mathcal{U}}(C_j)$. The reason for this difference is as follows.  In the first phase, the $0^{th}$-clustering contains $O(|A|)$ clusters (stars in this case), and hence it might take up to $O(|A|)$ rounds for each star leader to compute the list of all its neighboring stars. However, in phase $i\geq 2$, the $(i-1)^{th}$-level clustering contains $O(|A|^{1-1/k'})$ clusters, allowing stars to compute the exact number of neighboring clusters within $O(|A|^{1-1/k'})$ rounds, as will see next.  

\textbf{Phase $i = 1$}: In the first phase, all clusters are simply stars. 
Recall that $S_j, S_\ell$ are neighbors if $\Gamma(S_i)\cap S_j\neq \emptyset$. 
The neighboring stars $S_\ell$ of a star $S_j$ can be connected to the vertices of $S_j$ in one of two ways. Either a neighboring star $S_\ell$ is connected to $S_j$ by an edge with an endpoint at $a_j$, the center of the star $S_j$, alternatively $S_\ell$ is connected to $S_j$ by an edge with endpoint $b_{\ell'}\in S_j\setminus \{a_j\}$.

To approximate the star-degree of $S_j$, the leader $a_j$ computes the number of edges in each of these two types.  By adding these two values, we achieve a 2-approximation of the degree. 
To compute the edges of the first type, we let each vertex sending its star-ID to its neighbor in $G$. This allows the center $a_j$ to compute the number of stars in which it has a neighbor in. To compute the number of stars connected to $S_j$ via a vertex in $S_j \setminus \{a_j\}$, we do as follows: each vertex $a_i$ in $A$ sends a single ACK to each of its neighboring stars $S_\ell$ (by picking a unique neighbor in $S_j$ arbitrarily, and sending an ACK to it). The total number of ACKs received by the vertices in $S_j$ is exactly the number of stars that have a neighbor in $S_j \setminus \{a_j\}$. Finally, by letting $a_j$ add these two values, we have a 2-approximation for the star-degree. Note that this takes $O(1)$ communication rounds. Hence, in each iteration of the first phase, the star leaders can compute their update $\appdeg^{\mathcal{U}}(C)$ by applying the above computation restricted to the set 
$\mathcal{U}$ of unmarked stars.

\textbf{Phase $i \geq 2$}: At that phase, the clustering $\mathcal{C}_{i-1}$ contains $O(|A|^{1-(i-1)/k'})$ clusters. 
At the beginning of the phase, the vertex $a_j \in S_j$ computes a representative vertex $b_j \in S_j$ with a neighbor in $C'$ for every $C' \in \Gamma(S_j, \mathcal{C}_{i-1})$ (see Eq. (\ref{eq:neihinclustering})).  That is, for each neighboring cluster $C_{\ell}$, the center $a_j$ finds a vertex $b_{j,\ell} \in S_j$ that has a neighbor in $C_{\ell}$. This information is represented by a tuple $(b_{j,\ell}, C_{\ell})$ for every $C_\ell \in \Gamma(S_j, \mathcal{C}_{i-1})$. The length of this list is bounded by the number of clusters in $\mathcal{C}_{i-1}$, hence by $O(|A|^{1-(i-1)/k'})$. The leader $a_j$ sends this list to all its star members and by that, each vertex $b_{\ell}$ knows the set of clusters $C' \in \Gamma(S_j, \mathcal{C}_{i-1})$ that is responsible for. 
We next explain how the cluster centers can now compute their (exact) star-degree. 
Each representative vertex $b_{j,\ell}$ sends an ACK message to one neighbor in each of the clusters $C_\ell$ it is responsible for. Note that since the clusters are vertex-disjoint, each vertex $b_{j,\ell}$ sends only one message on each of its edges. In addition, note that in every cluster $C_\ell$, there is exactly one vertex that receives a message for each neighboring star $S_j$. Hence, by upcasting the number of ACKs received by the vertices of the cluster $C_{\ell}$, the center of the cluster $C_{\ell}$ can compute $\deg^{\mathcal{S}}(C)$ within $O(i)$ rounds. Since in the beginning of phase $i$, all stars are unmarked, $\deg^{\mathcal{S}}(C)$ is 
the initial unmarked star-degree of this phase. Overall, computing $\deg^{\mathcal{S}}(C)$ takes $O(i+|A|^{1-(i-1)/k'})$ rounds and hence we cannot recompute this value from scratch in each of the $O(|A|^{1-i/k'})$ iterations of that phase, as stars get marked.  

We next show how the cluster centers can maintain the unmarked star-degree of their cluster over the $O(|A|^{1-i/k'})$ iterations of this phase, when stars become marked. The goal is to have each star $S_j$ that got marked, send exactly one message to each neighboring cluster. Then, every cluster can count the total number of messages it received to know how many of its neighboring stars got marked and update its unmarked star-degree. It takes some care to make sure that exactly one message is sent to every neighboring cluster. The key will be to use the tuples $(b_{j,\ell}, C_{\ell})$ computed earlier, so for each neighboring cluster $C_{\ell}$, the corresponding vertex $b_{j, \ell}$ will be in charge of notifying the cluster $C_{\ell}$ if the star of $b_{j, \ell}$ gets marked. Since we have computed these tuples at the beginning of the phase, and each vertex in the star knows the neighboring clusters it is responsible for, we can ensure that each cluster gets exactly one message from each neighboring star that gets marked.

More formally, assume that the number of unmarked neighboring stars was maintained by each cluster up to the $(t-1)^{th}$ iteration and now consider the $t^{th}$ iteration where stars get marked. Let $\mathcal{C}'_{i-1}$ be the set of clusters whose centers have not yet joined $Z_i$. Note that all the current unmarked stars have all their cluster neighbors in $\mathcal{C}'_{i-1}$, as otherwise they would not be unmarked at that point. That is, for every unmarked star at the beginning of the  $t^{th}$ iteration, it holds that $\Gamma(S_j, \mathcal{C}'_{i-1})=\Gamma(S_j, \mathcal{C}_{i-1})$. 
At the end of the iteration, each star that becomes now marked sends a ``marking'' message to exactly one vertex $b_{\ell}$ in each of its neighboring clusters $C_{\ell} \in \Gamma(S_j, \mathcal{C}'_{i-1})$. For each neighboring cluster $C_\ell$, there is a tuple $(b_{j,\ell}, C_{\ell})$ which was computed before, where $b_{j, \ell}$ neighbors $C_{\ell}$. 
The vertices of each cluster $C_\ell$ in $\mathcal{C}'_{i-1}$ now upcast the number $x$ of ``marked'' messages they got, which allows the center of the cluster to subtract this number from its current unmarked star-degree in order to get an updated value for the unmarked star-degree.
\end{proof}
In addition, we have:
\begin{lemma}\label{ismaxima}
Every center in $Z'_{i-1}$ can verify if its cluster is an approximate-local-maxima (i.e., the value it computed for $(\appdeg^{\mathcal{U}}(C),ID(C))$ is larger than any other cluster it shares an unmarked star neighbor with) within $O(i)$ rounds.
\end{lemma}
\begin{proof}
The procedure is very similar to the one used in the standard algorithm $\NaiveSpanner$.
Each cluster has all of its vertices send $(\appdeg^{\mathcal{U}}(C),ID(C))$ to every neighboring star (either marked or not). 
Every vertex $b_j \in S_j$ then sends the largest tuple it received to the center $a_j$ of its star. The center of the stars $a_j \in S_j$ computes the largest tuple they received and sends this tuple to all the vertices in their $S_j$. These vertices then send an ACK to all their neighbors which originally sent this maximal message. If a cluster received a message from all its edges, then it is an approximate local maxima.
Since communication within clusters can be done in $O(i)$ rounds, the round complexity follows.
\end{proof}

Every cluster-center of an (approximate) local-maxima cluster joins $Z_i$ only if $\appdeg^{\mathcal{U}}(C)\geq |A|^{i/k'}$. We call such a cluster a successful cluster.
Every star $S_j$ that has a neighboring cluster which is successful marks itself.

\paragraph{(SII): Taking care of unclustered stars.}
Let $\mathcal{U}'$ be the remaining unmarked stars and let $\mathcal{C}'_i$ be the remaining clusters in $\mathcal{C}_{i-1}$ whose centers did not join $Z_i$. 
At this point, we let each cluster $C \in \mathcal{C}'_i$ add one edge to each of its neighboring unmarked stars in $\Gamma^{\mathcal{U}'}(C)$.
\begin{lemma}\label{SIIrounds}
Step (SII) takes $O(i+|A|^{i/k'})$ rounds.
\end{lemma}
\begin{proof}
Each of the remaining clusters in $\mathcal{C}'_i$ has $O(|A|^{i/k'})$ unmarked stars in $\mathcal{U}'$. Hence, the cluster center can compute a list of its neighboring unmarked stars in $O(|A|^{1/k'})$ rounds. It can also, within $O(|A|^{i/k'})$ rounds, compute for each unmarked star, a neighbor in its cluster, and thus create a list of tuples $(S_j, v)$, where for each neighboring unmarked star $S_j$ there is a vertex $v \in C$ in the cluster such that $v$ neighbors $S_j$. The cluster center sends this list of tuples to the entire cluster in $O(i+ |A|^{1 - 1/k'})$ rounds. As a result, every vertex $v \in C$ that is assigned to $S_j$, adds an edge to one neighbor in $S_j$
\end{proof}

\paragraph{(SIII): Defining the clusters of $\mathcal{C}_i$.}
This is done as in Baswana-Sen -- building a BFS tree up to depth $i$, only in the star-graph, while breaking ties based on distance and cluster-ID. This can be done by running a BFS on the star-graph, centered at each cluster center. When a star joins a BFS tree, its center informs all of the other vertices in the star of the cluster which was joined. In the case where there is a tie, and two BFS trees reach the same star at the same depth, the star will choose which of the clusters to join based on the cluster IDs.

We now describe the next phase of the algorithm, which is done after the $(k'-1)^{th}$ level of clustering has been computed.

\paragraph{Last phase:}
After applying the above for $(k'-1)$ phases, in the $(k'-1)^{th}$-level of clustering, we have $O(|A|^{1 - 1/k'})$ clusters of radius $(k'-1)$ in the star-graph. At this point, each star $S_j$ adds to $H$ an edge to one of its neighbors in each of its neighboring clusters. This completes the description of the algorithm. 

We now proceed to analyze the number of rounds and provide a stretch and size analysis for the output spanner.

\paragraph{Size Analysis.}
We show that the output spanner contains $O(k |A|^{1+1/k'} + |B|)$ edges. Initially, when each vertex of $B$ joins a star centered at $A$, we add $|B|$ edges to the spanner.

We now bound the number of edges added when constructing the $i^{th}$ clustering $\mathcal{C}_i$ from the $(i-1)^{th}$ clustering $\mathcal{C}_{i-1}$. Every cluster of the $(i-1)^{th}$-clustering that joins the $i^{th}$-clustering has a unmarked star-degree at least $|A|^{i/k'}$ at the time it joined. Since the unmarked star-neighborhoods of the elected clusters are disjoint (since when a cluster is chosen, it marks its neighbors), there are at most $|A|^{1 - i/k'}$ clusters whose centers join $i^{th}$-clustering. Also, there are $|A|^{1 - i/k'}$ iterations for electing the cluster of locally maximal unmarked star-degree. 
Let $\mathcal{C}'_{i-1}$ be the remaining clusters whose cluster centers did not join $Z_i$ and let $\mathcal{U}'$ be the remaining unmarked stars. It then holds that $\deg^{\mathcal{U}'}(C) \leq |A|^{i/k'}$.

Each one of these non-chosen clusters adds to $H$ one edge to each of its unmarked star-neighbors. Since there are $O(|A|^{i/k'})$ such star neighbors for each cluster, and $O(|A|^{1 - (i-1)/k})$ clusters, overall, this adds $O(|A|^{i/k'}|A|^{1 - (i-1)/k'}) = O(|A|^{1 + 1/k'})$ edges to the spanner. Summing over all $k'$ phases give a total of $O(k'|A|^{1 - 1/k'})$ edges that were added to the spanner.

At the end of these $k'$ phases, the $(k-1)^{th}$-clustering has $O(|A|^{1 - (k'-1)/k'}) = O(|A|^{1/k'})$ clusters. For each such cluster, the partial spanner contains a BFS tree of radius $(k' - 1)$, which adds $O(|V|) = O(|A| + |B|)$ edges. At this point, we let each vertex in $A$ add an edge to each of its neighboring clusters. In total, this adds $O(|A|)\cdot O(|A|^{1 - (k'-1)/k'}) = O(|A|^{1 + 1/k'})$ edges to the spanner.

Summing over all of the above gives the desired bound of $O(k'|A|^{1+1/k'}+|B|) = O(k|A|^{1+2/k}+|B|)$.

\paragraph{Stretch Analysis.} Consider an edge $(v_a, v_b) \in A \times B$. Either at some point one of $v_a$ or $v_b$ got unclustered, or both were clustered throughout. First, we deal with the case where at least one of $v_a$ or $v_b$ became unclustered. 
Assume without loss of generality, that $v_a$ stopped being clustered for the first time at step $i$ \emph{not after} $v_b$ (in the case where $v_b$ became unclustered first, the argument proceeds identically). That is, both the stars of $v_a$ and $v_b$ are clustered in the $(i-1)^{th}$-clustering $\mathcal{C}_{i-1}$, but the star of $v_a$ is unclustered in $\mathcal{C}_{i}$. The cluster of $v_b$ in the $(i-1)^{th}$-clustering added one edge to each neighboring unclustered star. In particular, it adds an edge to the star containing $v_a$. Since the cluster of $v_b$ has radius $i-1$ in the star-graph, the distance from $v_b$ to $v_a$ in the star-graph is at most $(2i - 1)$. Since $i \le k'$, in the original graph their distance is at most $4k' - 1 = 2k - 1$. 

We now deal with the remaining case where $v_a$ and $v_b$ are both clustered throughout. So, in particular, $v_a$ and $v_b$ are in the $(k'-1)^{th}$ clustering.
after the $(k'-1)^{th}$ clustering, each star neighboring the cluster of $v_a$ adds an edge to the cluster of $v_a$. In particular, the star of $v_b$ will add an edge to the cluster of $v_a$. Hence, the distance from the star of $v_a$ to the star of $v_b$ in the spanner induced on the star-graph will be at most $2k' - 1$, and therefore in the original graph their distance will be at most $2(2k' - 1) + 1 = 2k - 1$, as desired.


\paragraph{Round Complexity.} 
We now turn to bound the number of rounds of the algorithm. 
Consider the $i^{th}$-step, where we construct the $i^{th}$-clustering. We will analyze the number of rounds taken by Step (SI), (SII), and (SIII) separately. Steps (SI) and (SII) take the most time, with (SIII) being faster in comparison.

First, we analyze the round complexity of Step (SI). For each cluster in the $(i-1)^{th}$ clustering, computing its unmarked star degree, the number of neighboring unmarked stars, takes $O(\min(|A|^{1-1/k'}, |A|^{1 - (i-1)/k'}))$ rounds by Lemma \ref{unmarkeddeg}. Summing over all $i$ from $1$ to $k'-1$ gives $O(|A|^{1 - 1/k'}) = O(|A|^{1 - 2/k})$. Then, to check if a cluster is a local maxima takes $O(i)$ rounds by Lemma \ref{ismaxima}.  

Step (SII) takes a total of $O(i+|A|^{i/k'})$ rounds by Lemma \ref{SIIrounds}. Summing over all $i$ from $1$ to $k' - 1$ gives $O(k'^2 + |A|^{1- 1/k'}) = O(k^2 + |A|^{1- 2/k})$ rounds. If $k^2$ is comparable with $|A|^{1- 2/k}$, then we can simply add a spanner with $k' = 10$, thus having an edge bound of $O(|A|^{1 + 1/10} +|B|) \le O(k|A|^{1+1/k} + |B|)$ edges in the spanner, as desired. Furthermore, the number of rounds required will be $O(|A|^{1- 1/10}) \le O(|A|^{1 - 2/k})$. Thus, we can ignore the additive factor of $k^2$, achieving a bound of $O(|A|^{1- 2/k})$ rounds.

Step (SIII) takes time $O(i)$, as in Baswana-Sen.

This completes the proof of Lemma \ref{bipartite}. 
For odd $k$, we get a similar algorithm and analysis. The only difference is that instead of constructing a $(k-1)$-spanner of the star-graph, we construct a $(k-2)$-spanner of the star-graph. This results in:
\begin{lemma}[Bipartite Spanners: Odd k]
Let $G = (A \cup B, E)$ be a bipartite graph.Then one can construct a $(2k - 1)$-spanner with $O(k|A|^{1+2/(k-1)} + |B|)$ edges within $O(|A| + k^2)$ rounds in the \congest\ model for odd $k \ge 3$. 
\end{lemma}

\subsection{Key Tool (II): Superclustering -- Grouping Baswana-Sen Clusters}
\label{sec:supercluster}
\paragraph{Why Superclusters?}
In this section, we describe the main tool that allows us to speed up Alg. $\NaiveSpanner$ by a factor of $\sqrt{n}$. The idea is to group the $n^{1-i/k}$ clusters in the $i^{th}$-clustering $\mathcal{C}_i$ into $\sqrt{n}$ superclusters, each containing $O(n^{1/2-i/k})$ clusters. Then, instead of iterating over clusters one by one (as in Alg. $\NaiveSpanner$), we iterate over the superclusters. Each time, either \emph{all} the cluster centers of a given supercluster join the next level of clustering, or none of them join. As will be shown later, in order to construct the $i^{th}$-clustering $\mathcal{C}_{i}$, it will be sufficient for our algorithm to consider $n^{1/2-1/k}$ superclusters (and not all $\sqrt{n}$ superclusters), hence yielding the round complexity of $O(n^{1/2-1/k})$ (for fixed $k$).
For a supercluster to compute the number of its (unmarked) neighbors, all cluster centers in a given supercluster should be able to communicate efficiently. For that purpose, we make sure that the cluster centers in each supercluster are connected by an $O(2^k)$-depth tree\footnote{This bound arises in Appendix \ref{append:largediameter} and will be discussed later on.}, and that the trees of different superclusters are edge-disjoint. These trees will allow us to aggregate information to leader of each supercluster in parallel, see Fig. \ref{fig:smalllarge}.
\begin{figure}[h!]
\begin{center}
\includegraphics[scale=0.25]{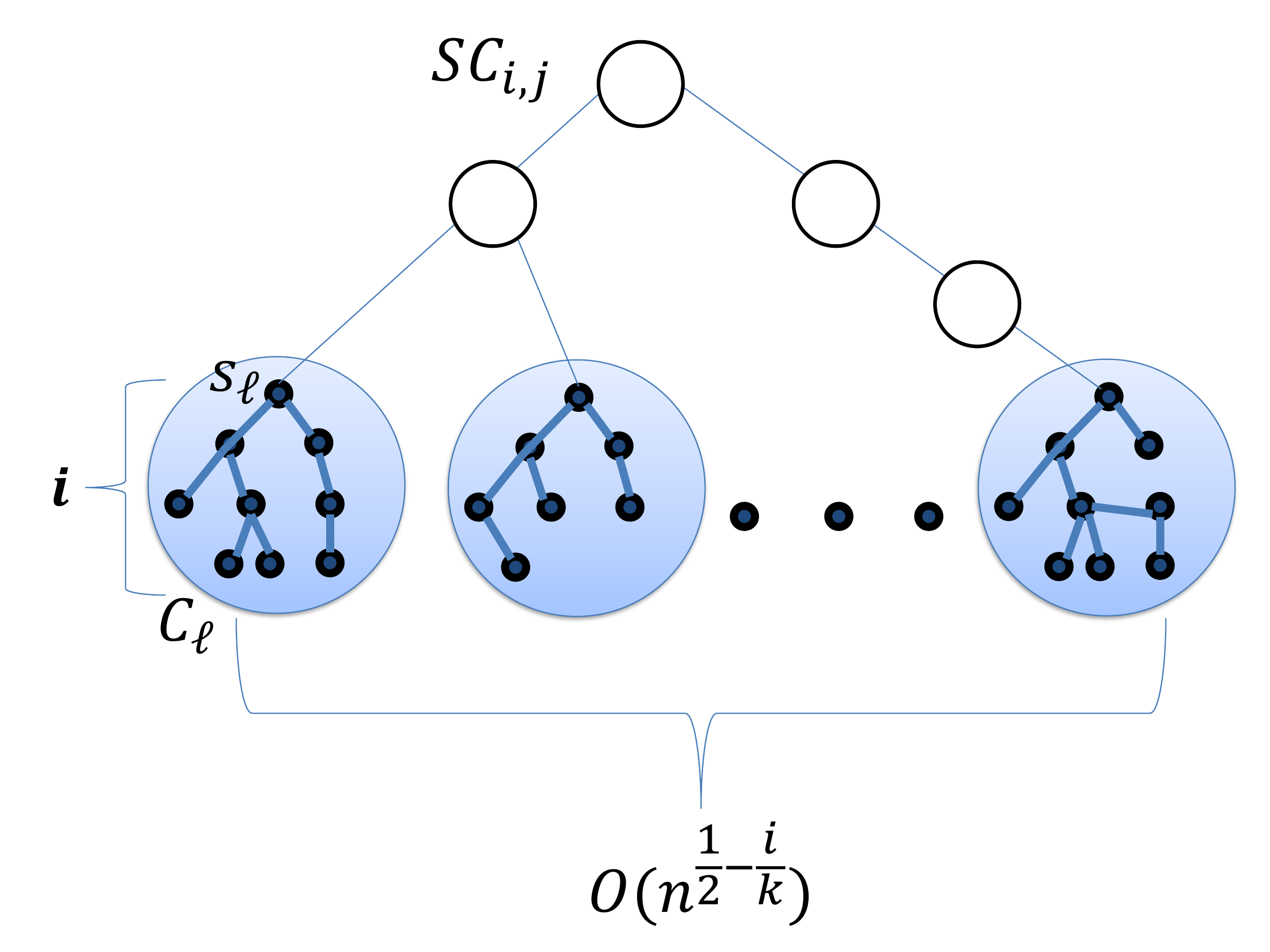}
\includegraphics[scale=0.25]{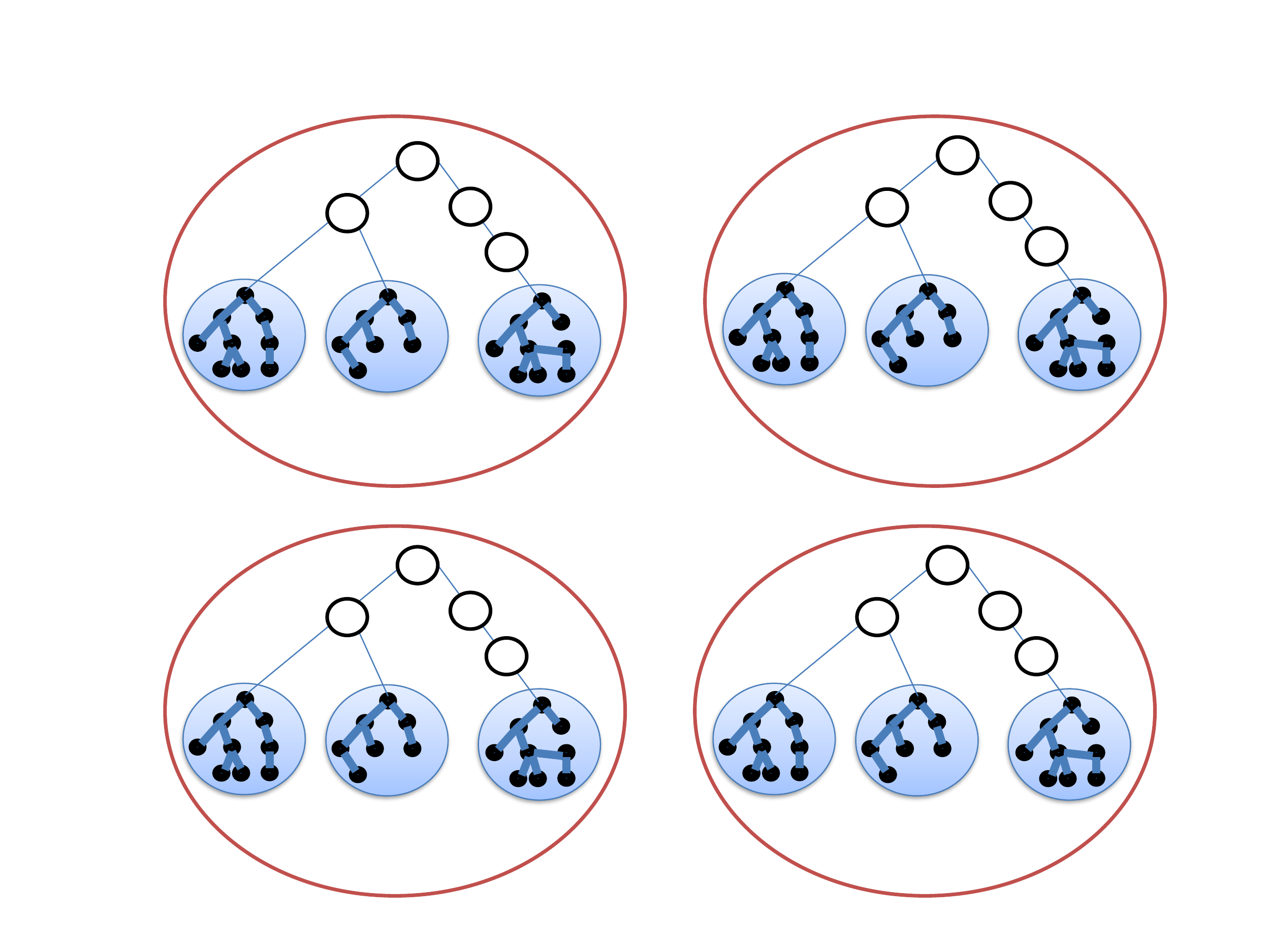}
\caption{
\label{fig:smalllarge} Illustration of the superclusters. Left: Baswana-Sen clusters are shown in blue circles. 
Each such cluster $C_\ell$ contains a center $s_\ell$ and a BFS tree of depth $i$ rooted at $s_\ell$ spans all the vertices of the cluster. The supercluster $\supercluster_{i,j}$ groups these clusters by adding a low-diameter tree that connects all the cluster centers. It might be the case that the vertices connecting the centers are vertices that belong to clusters; either in the same supercluster or in another one. Since the trees are edge disjoint, this does not increase congestion (by more than factor 2) in any case. Right: The collection of all $O(n^{1-i/k})$ clusters in the $i^{th}$-level of the clustering grouped into $O(\sqrt{n})$ superclusters.}
\end{center}
\end{figure}

\paragraph{Defining the Superclusters.}
Let $\mathcal{C}_i$ be a collection of $O(n^{1-i/k})$ $i$-clusters. 
A \emph{supercluster} $\supercluster_{i,j}=\left\{C_{j_1}, \ldots, C_{j_\ell}\right\}$ is a collection of clusters from $\mathcal{C}_i$. A \emph{Superclustering} $\superclustering_i=\{\supercluster_{i,1},  \ldots, \supercluster_{i,p}\}$ is a covering partition of all clusters from $\mathcal{C}_i$. That is, $\bigcup_{j=1}^p\supercluster_{i,j}=\mathcal{C}_i$, and the superclusters are cluster-disjoint (every cluster in $\mathcal{C}_i$ belongs to exactly one supercluster).
To select the cluster centers of level $i$, the algorithm constructs in each phase $i \in \{1, \ldots, k/2\}$ a superclustering $\superclustering_i$ which satisfies some helpful properties. We call a superclustering satisfying these properties a \emph{nice superclustering}. Before defining the properties of a nice supercluster, we introduce some notation.
For a supercluster $\supercluster_{i,j}=\left\{C_{j_1}, \ldots, C_{j_p}\right\}$, let $V(\supercluster_{i,j})=\bigcup_{C \in \supercluster_{i,j}}C$ be the set of all vertices in its clusters and 
$N_V(\supercluster_{i,j})=|V(\supercluster_{i,j})|$ be the number of vertices in the supercluster $\supercluster_{i,j}$. Also, let $N_C(\supercluster_{i,j})$ denote the number of clusters that the supercluster $\supercluster_{i,j}$ contains.
A supercluster $\supercluster_{i,j}$ with only one cluster (i.e., $N_C(\supercluster_{i,j})=1$) is called a \emph{singleton}. In addition, a singleton supercluster is called a \emph{small-singleton} if $N_V(\supercluster_{i,j})\leq \sqrt{n}$ (otherwise, if $N_V(\supercluster_{i,j})>\sqrt{n}$, it is a \emph{large-singleton}).  
Our $(2k-1)$-spanner construction is based upon the construction of superclusters with some \emph{nice} useful properties, as defined next. 

\paragraph{Nice Superclustering.}
A superclustering $\superclustering_i=\{\supercluster_{i,1},  \ldots, \supercluster_{i,\ell}\}$ is \emph{nice} if it contains $\ell=O(\sqrt{n})$ superclusters, and each of these superclusters $\supercluster_{i,j}\in \superclustering_i$ satisfies the following:
\begin{description}\label{desc:nice}
\item{(N0) [Singleton]}
If $N_V(\supercluster_{i,j})=\Omega(\sqrt{n})$, then $N_C(\supercluster_{i,j})=1$.
\end{description} 
Every non-singleton supercluster $\supercluster_{i,j}$ (i.e., every supercluster containing at least two clusters) satisfies:
\begin{description}
\item{(N1) [Cluster Balance]}
$N_C(\supercluster_{i,j})=O(n^{1/2-i/k})$, and 
\item{(N2) [Vertex Balance]}
$N_V(\supercluster_{i,j})=O(\sqrt{n})$.
\item{(N3) [Connectivity]}
In the graph $G$, each $\supercluster_{i,j} \in \superclustering_i$ has a tree $T(\supercluster_{i,j})$ of depth\footnote{When the diameter of the original graph $G$ is $O(1)$, the diameter of $T(\supercluster_{i,j})=O(1)$. The term $O(2^{k})$ appears when dealing with graphs of large diameter, and will become clear in Section \ref{append:largediameter}.} $O(2^{k})$. In addition, the trees $T(\supercluster_{i,1}), \ldots, T(\supercluster_{i,\ell})$ are edge-disjoint. 
\end{description}
\dnsparagraph{Intuitive discussion of these properties}
Property (N0) implies that if a supercluster has many vertices (more than $\sqrt{n}$), then it is a singleton supercluster.
Property (N1) implies that non-singleton superclusters with at least two clusters are balanced with respect to the number of \emph{clusters} from $\mathcal{C}_i$ that they contain. Since there are $O(n^{1-i/k})$ clusters in the $i^{th}$ clustering, dividing it ``fairly" between $\sqrt{n}$ superclusters yields this bound. 
Property (N2) also implies a balance among non-singleton superclusters, but this time with respect to the number of vertices. Finally, Property (N3) provides the existence of a $O(2^k)$-depth tree that connects the cluster centers of that supercluster. This ``weird" looking depth of $O(2^k)$ shows up when computing the $0^{th}$-level superclustering for general graphs (for graphs of constant diameter a much simpler construction exists). In particular, it shows up in Step (SI) of Alg. $\ConstructZeroSuper$ described in Appendix \ref{append:largediameter}. Finally, (N4) requires these trees to be edge-disjoint to allow communication within different superclusters, \emph{in parallel} without congestion.

As will be shown in the next subsection, to satisfy Properties (N1) and (N2), the construction of the $i^{th}$-level of superclustering requires to partition both the vertices and the clusters into \emph{balanced} the $\sqrt{n}$ superclusters. The next lemma describes the key tool to achieve it. 

\paragraph{The Balanced Partitioning Lemma.}
The input to the partitioning lemma is a vertex-weighted tree $T$, where every vertex $v$ in $T$ has a non-negative weight $w(v)$ and in addition, we are given a bound $B$ on the allowed total weight of each tree.
The goal is to partition the tree into 
edge-disjoint subtrees, such that, all but one of the subtrees have a weight in $[B,2B]$.
The lemma achieves this but with some subtle specification. It partitions the vertices of the tree $T$ into $p$  disjoint sets: $\widehat{V}(T_0), \widehat{V}(T_1), \ldots, \widehat{V}(T_p)$. The total weight of each set $\widehat{V}(T_i)$, except for at most one, $\widehat{V}(T_0)$, is bounded by $[B,2B]$. Hence, the partition respects the weight bound. Next, each set $\widehat{V}(T_i)$ is connected by a subtree $T_i \subseteq T$. The important feature of these trees $T_i$ is that they might contain an additional vertex $v \in V(T) \setminus \widehat{V}(T_i)$. This additional vertex $v$, if exists, is the root of $T_i$ and it is essential to connect the vertices in $\widehat{V}(T_i)$. Intuitively, this additional vertex helps us to communicate between the vertices of $\widehat{V}(T_i)$. Even though the trees $T_i$ are not vertex disjoint, they are shown to be edge disjoint, which is sufficient for our applications. 
\begin{lemma}[Balanced Partitioning Lemma]\label{lem:balancedpartitioning}
\label{thm:partition}
In $O(\diam(T))$ rounds, one can construct subtrees $T_0, T_1,\ldots, T_p \subseteq T,$ with roots $r(T_0), r(T_1), \ldots r(T_p)$ and corresponding \emph{disjoint} vertex sets $$\widehat{V}(T_0), \widehat{V}(T_1), \ldots, \widehat{V}(T_p)$$ such that: 
\begin{description}
\item{(D1)}
The $\widehat{V}(T_i)$ sets are vertex disjoint and $\bigcup_{i=1}^p\widehat{V}(T_i)=V(T)$.
\item{(D2)}
$W(T_i)\in [B,2B]$ for every $i \geq 1$, and $W(T_0)\leq 2B$ where $W(T_i)=\sum_{u \in \widehat{V}(T_i)}w(u)$.
\item{(D3)}
$V(T_i)=\widehat{V}(T_i)\cup r(T_i)$.
\item{(D4)}
All $T_0, \ldots, T_p$ are edge-disjoint and with diameter at most $\diam(T)$.
\end{description}
\end{lemma}
Intuitively, the important vertex set of the tree $T_i$ is the set of vertices $\widehat{V}(T_i)$ and hence the weight of the tree in Property (D2) is defined by summing over all these vertices (instead of summing over all vertices in the tree). Property (D3) implies that the tree $T_i$ might contain, in addition to $\widehat{V}(T_i)$, also an additional vertex -- its root -- that allows the connectivity of the set $\widehat{V}(T_i)$ in $T_i$. 
The full proof of Lemma \ref{lem:balancedpartitioning} appears Appendix \ref{append:partition}.
In the common application of this lemma, the tree $T$ is a tree that connects the cluster-centers of a given supercluster, these cluster-centers are given a weight (e.g., the size of their cluster) and the remaining vertices in $T$ are given a zero weight. The bound corresponds to the maximum allowed number of clusters (or vertices) in the supercluster (as in \ref{desc:nice}(N2,N3)). 

\def\APPENDBALANCEPART{
\begin{proof}
The idea of the proof is as follows: starting with root of the tree, we recursively partition each of its children subtrees. The challenge is in combining the resulting partitions. The main obstacle is that if each of the subtrees has a leftover set in its output partition, it is not clear how to combine the child-partitions at the root in a way that results in only one leftover set. We deal with that problem by merging leftover sets together. Since we only want the trees to be edge-disjoint, we allow the trees to share one vertex in common (e.g., the root). 

Let $r_1, r_2, \ldots, r_d$ be the children of the root $r$ in the tree $T$.
Our algorithm is recursive. For the recursion to hold, we need an additional auxiliary property:
\begin{description}
\item{(D5)}
$W(T_0)\leq 2B$ and $r(T_0)=r(T)$. If $W(T_0) \le B$ (we call such a $T_0$ the \textit{leftover set}), then $\widehat{V}(T_0)=V(T_0)$.
\end{description}

For the base case (if the tree contains only one vertex), we simply return the trivial partition. 
Otherwise, we recursively apply the algorithm to each of the subtrees of $r_1, \ldots, r_d$. We now wish to combine these $d$ partitions to create a single balanced partition. Note that the union of these partitions satisfies Properties (D3) and (D5). Also, it satisfies Property (D1) for all vertices except for the root $r$ of $T$.
Property (D2) is not yet satisfied since there might be small leftover sets in each of the recursive partitions. Therefore, when combining the partitions we may have up to $d$ small sets. In order to deal with this, we combine leftover sets together so that there will be at most one small leftover set (all other sets will be combined so their total weight is at least $B$).

Suppose the leftover sets $S_1, S_2, \ldots, S_d$ have weights $$W(S_1), W(S_2), \ldots, W(S_d) \le B.$$ 
We iterate over these subsets in order and group them together until the combined set $S'_1$ has weight larger then $B$ (i.e., $S'_1$ is the union of $S_1, S_2, \ldots, S_i$, for the smallest $i$ such that the combined weight is larger than $B$). We then proceed combining the sets $S_{i+1},...S_{i+j}$, until we get another set $S'_2$ of weight at least $B$.
By that, we get a collection of new sets $S'_1, \ldots, S'_\ell$, where each except perhaps for the last one, $S'_\ell$, has weight in $[B,2B]$. 
We then add the root $r$ to the trees $T'_{j}$ of each of the $S'_j$ for $j \in \{1, \ldots, \ell-1\}$ but only as an auxiliary vertex to aid the communication. If the weight of the last set $S'_{\ell}$ is at least $B$, $r$ is added to $S'_{\ell}$ as an auxiliary communication vertex, and $r$ becomes a singleton leftover supercluster. Otherwise, if the weight of last set $S'_{\ell}$ is less than $B$, The root $r$ is added to the last set $S'_{\ell}$. Note that in this case, even when adding $r$ to $\widehat{V}(T'_{\ell})$, we still have a total weight less than $2B$ since $w(r)\leq B$. If the weight of $S'_\ell$ is at least $B$, we let $r$ be its own set $S'_{\ell+1}$. This guarantees property (D1) because all vertices other than $r$ are in some $\hat{V}$ by the recursion, and we ensure that $r$ is in one of the sets as well. Note that this also provides property (D2) for each of the $S'_{j}$ sets, and for the other sets (D2) holds by the recursion. 

By this construction only the last set of the $S_i$ can have weight less than $B$, providing property (D2). Since only the root is common to the $S'_{j}$ sets, the resulting trees are indeed edge disjoint, thus satisfying property (D4).
 
We now consider the distributed implementation. Note that at each step of the recursion, the root $r$ of the current subtree has all the information needed to decide on how to combine the leftover subsets. For each child $v$ whose subtree has a leftover subset $T_0$, the root $r$ sends to $v$ the ID of the new vertex set $T_i$ which $T_0$ combined with. The child $v$ can then propagate that information downward to the rest of the nodes in $T_0$ leftover subset, by that all the vertices know the tree $T_i$ to which they belong.
\end{proof}
}
%

\subsection{$(2k-1)$-Spanners in $O(2^{k}\cdot n^{1/2-1/k})$ Rounds} 
\label{sec:alg}
We first consider the construction for graphs with constant diameter. At the end of the section, and in Section \ref{append:largediameter} we discuss the extension for general graphs with arbitrary diameter.
Recall that for $i \leq k/2$, $\mathcal{C}_{i}$ is a clustering that contains $O(n^{1-i/k})$ vertex-disjoint clusters centered at the vertices $Z_i$. The set of $i$-clustered vertices $V_i$ are in $\Gamma_i(Z_i)$.

The first part of the algorithm contains $k/2$ phases. In each phase $i\in \{1, \ldots, k/2\}$, we are given a $(i-1)^{th}$ nice superclustering $\superclustering_{i-1}$ (whose superclusters contain the clusters of $\mathcal{C}_{i-1}$) and the current spanner $H$. We then construct the $i^{th}$ nice superclustering $\superclustering_i$ and add edges to $H$ in order to take care of the newly unclustered vertices in $V_{i-1}\setminus V_i$. At the end of the first part, we have a $(k/2)^{th}$ superclustering $\superclustering_{k/2}$ with $O(\sqrt{n})$ clusters. At that point, the number of clusters is small enough, and so Alg. $\NaiveSpanner$ can be applied.

\paragraph{Constructing the $0^{th}$-level superclustering $\superclustering_{0}$ in $O(\diam(G))$ rounds.}
To compute $\superclustering_{0}$, we apply the Partitioning Lemma \ref{lem:balancedpartitioning} on a BFS tree $T$ rooted at some arbitrary vertex (e.g., of maximum ID) using weights of $w(v)=1$ for each $v \in V$ and bound $B=O(\sqrt{n})$. This partitions the vertices into $\Theta(\sqrt{n})$ subsets $S_{i}$, each of size $O(\sqrt{n})$. Each such subset $S_i=\{v_{i,0}, \ldots, v_{i,\ell}\}$ defines a supercluster $\supercluster_{0,i}=\{\{v_{i_0}\}, \ldots, \{v_{i,\ell}\}\}$ containing the singleton clusters of $S_i$'s vertices.  
By that, we get $O(\sqrt{n})$ superclusters $\superclustering_{0}=\{\supercluster_{0,1}, \ldots, \supercluster_{0,\sqrt{n}}\}$. By the partitioning lemma, we also have a tree $T_i$ for each $\supercluster_{0,i}$, satisfying Prop. (N3).  
%

\paragraph{The $i^{th}$ phase of Algorithm $\ImprovedSpanner$ for $i\in \{1, \ldots, k/2\}$.}
At the beginning of the phase, we are given the $(i-1)^{th}$-clustering $\mathcal{C}_{i-1}$ grouped into 
the nice superclustering $\superclustering_{i-1}$. 
Our first goal is to use the superclustering $\superclustering_{i-1}$ to define the set of new $O(n^{1-i/k})$ cluster centers $Z_i$. The high-level idea here is to implement Alg. $\NaiveSpanner$ on each supercluster rather than on each cluster. Given a set $U$ of unmarked vertices and a supercluster $\supercluster \in \superclustering_{i-1}$, define its \emph{unmarked neighborhood} and \emph{unmarked degree} by
\begin{equation}
\label{eq:lowexp}
\Gamma^{U}(\supercluster)=\bigcup_{v \in V(\supercluster)}(\Gamma^+(v)\cap U) \mbox{~~and~~} \deg^U(\supercluster)=|\Gamma^{U}(\supercluster)|~.
\end{equation}
%
Similarly to before, we say that a supercluster $\supercluster$ is a \emph{local-maxima} in its unmarked neighborhood, if for every other $\supercluster'$ such that  
$\Gamma^{U}(\supercluster)\cap \Gamma^{U}(\supercluster')\neq \emptyset$, it holds that $$(\deg^U(\supercluster),ID(\supercluster))>(\deg^U(\supercluster'),ID(\supercluster')).$$ 

We say that supercluster $\supercluster$ has \emph{low-expansion} if $\deg^U(\supercluster)\leq n^{1/2+1/k}$. Otherwise, it has \emph{high-expansion}. Note that unlike the previous algorithms presented before, here the expansion threshold $n^{1/2+1/k}$ is independent\footnote{The intuition is that in each level $i$, the superclusters have at most $\sqrt{n}$ vertices, and we say that it has high expansion if the size of its neighborhood size is factor $n^{1/k}$ larger.}  of the level $i$. See Table \ref{table:notation} for a summary of notation.
\begin{table}[h!]
\begin{tabular}{{c}{l}r}
Notation              & Meaning & Page\\
\hline
$\mathcal{C}_i$  & {\shortstack{The $i$-level of clustering containing $O(n^{1-i/k})$ clusters}}
& \pageref{sec:kspanner}\\

$Z_i$ & The cluster centers of $\mathcal{C}_i$.  & \pageref{sec:kspanner}\\ 

$V_i$   & The vertices that are $i$-clustered, i.e., $V_i=\Gamma_i(Z_i)$       & \pageref{sec:kspanner}\\

$\superclustering_i$ & {\shortstack{Superclustering at level $i$: \\A partition of the $O(n^{1-i/k})$ clusters of $\mathcal{C}_i$ into $\sqrt{n}$ groups}} &
\pageref{sec:supercluster} \\ 

$\supercluster_{i,j}$ & {\shortstack{Supercluster in $\superclustering_i$ containing a subset of $O(n^{1/2-i/k})$ clusters \\belonging to $\mathcal{C}_i$}} &
\pageref{sec:supercluster} \\ 

$\Gamma^U(C), \deg^U(C)$ & The unmarked neighborhood, degree of cluster $C$ & \pageref{sec:naive} \\ 

$\Gamma^U(\supercluster), \deg^U(\supercluster)$ & The unmarked degree of supercluster $\supercluster$ &
\pageref{sec:alg} \\ \\
Local-maxima (super)cluster  & {\shortstack{(Super)cluster that has the maximum unmarked degree \\in its unmarked neighborhood, breaking ties based on ID}}
&
\pageref{sec:naive},\pageref{sec:alg} \\
Successful (super)cluster   & {\shortstack{Local-maxima supercluster (cluster) with large unmarked degree}}
&
\pageref{sec:naive},\pageref{sec:alg} \\
\end{tabular}
\caption{List of notation}
\label{table:notation}
\end{table}

\paragraph{Step (S1) of phase $i$: Selecting the centers $Z_i$.}
Selecting the $O(n^{1-i/k})$ cluster centers of $Z_i$ is done in $O(n^{1/2-1/k})$ iterations. We start by unmarking all vertices. At each iteration, we have a set $U$ of remaining unmarked vertices and a subset of remaining superclusters $\superclustering'_{i-1}$ of superclusters whose cluster centers have not yet been added to $Z_i$. All superclusters $\supercluster \in \superclustering'_{i-1}$ compute their unmarked degree $\deg^U(\supercluster)$ in parallel. (This can be done in $O(i\cdot 2^{k})$ rounds thanks to Prop. (N3) in Desc. \ref{desc:nice}). 
\begin{definition}[Successful Supercluster]
A supercluster $\supercluster$ that is local-maxima in its unmarked neighborhood and has high-expansion, that is  $\deg^U(\supercluster)\geq n^{1/2+1/k}$, is called a \emph{successful} supercluster.
\end{definition}
It is easy to see that the leader (vertex $v$ of maximum ID in $V(\supercluster)$) of every supercluster $\supercluster$ can verify in $O(i\cdot 2^{k})$ rounds whether it is a local-maxima in its unmarked neighborhood.

In the algorithm, each successful supercluster $\supercluster$ adds all its cluster centers to $Z_i$, and mark all the vertices in $\Gamma^{U}(\supercluster)$. This continues for $O(n^{1/2-1/k})$ iterations. 

As we will show in the analysis section, since the unmarked neighborhoods of successful superclusters are disjoint and large, there are at most $O(n^{1/2-1/k})$ such superclusters. In addition, by Prop. (N2), each supercluster has $O(n^{1/2-(i-1)/k})$ clusters, overall $|Z_i|=O(n^{1/2-(i-1)/k} \cdot n^{1/2-1/k})=O(n^{1-i/k})$ as desired.

\paragraph{Step (S2) of phase $i$: Taking care of unclustered vertices.}
After $O(n^{1/2-1/k})$ steps of computing successful superclusters, all remaining superclusters $\supercluster$ have low-expansion with respect to the remaining unmarked vertices $U'$. That is, $\deg^U(\supercluster)\leq n^{1/2-1/k}$. First, we take care of the singleton superclusters.\\ \\
\textbf{(S2.1): Singleton supercluster $\supercluster$ with low-expansion:}
Each unmarked vertex $u \in U'$ add to $H$ an edge to one of its neighbor in $\Gamma(u)\cap V(\supercluster)$. Since there are $O(\sqrt{n})$ superclusters, each with $\deg^{U'}(\supercluster)\leq n^{1/2+1/k}$, overall we add $O(n^{1+1/k})$ such edges. \\\\
\textbf{(S2.2): Non-singleton superclusters $\supercluster$ with low-expansion:}
Here, the construction of sparser spanners for bipartite graphs comes into play (see Sec. \ref{sec:unbalancedbipartite}). 
Recall that by Prop. (N2), $N_V(\supercluster)=O(\sqrt{n})$ vertices. 
Let $\Gamma^{U',-}(\supercluster)=\Gamma^{U'}(\supercluster)\setminus V(\supercluster)$ be the unmarked neighbors of $\supercluster$ excluding the vertices of the supercluster $\supercluster$. Since $\supercluster$ has low-expansion, it also holds that  $|\Gamma^{U',-}(\supercluster)|=O(n^{1/2-1/k})$. For every such supercluster $\supercluster$, we consider 
the bipartite graph $B(\supercluster)=(V(\supercluster),\Gamma^{U',-}(\supercluster))$, and apply Alg.
$\BipartiteSpanner$ to compute for it a $(2k-1)$-spanner $H(\supercluster)\subseteq B(\supercluster)$ with $O(n^{1/2+1/k})$ edges (see Lemma \ref{bipartite}). This is done for all the graphs $B(\supercluster)$ in parallel.

Note that the graphs $B(\supercluster)$ are not necessarily vertex disjoint since an unmarked vertex can appear in several such graphs. The key observation that allows the parallel computation of all these spanners, is that every edge $(u,v)$ can belong to at most \emph{two} bipartite graphs, say, $B(\supercluster)$ and $B(\supercluster')$, where $\supercluster,\supercluster'$ is the supercluster of $u,v$ respectively\footnote{Recall that the superclusters share no vertex in common.}. 
Overall, since there are $O(\sqrt{n})$ superclusters, this adds $O(n^{1/2} \cdot n^{1/2+1/k})=O(n^{1+1/k})$ edges.

Finally, it remains to take care of all edges between vertices belonging to the \emph{same} supercluster. Note that in Alg. $\NaiveSpanner$, there was no need for such a step since all vertices belonging to the same cluster are connected in $H$ by an $i$-depth BFS tree rooted at the cluster center. However, in our setting, vertices that belong to \emph{different} clusters of the \emph{same} superclusters might still have large stretch (as cluster centers of the same supercluster might be at distance $O(2^k)$ in $G$). At that point, we use the fact that all superclusters are vertex disjoint and each contains $O(\sqrt{n})$ vertices. We then recursively apply the algorithm $\ImprovedSpanner$ on each of these superclusters in parallel. That is, we apply $\ImprovedSpanner$ on the induced subgraph on $V(\supercluster)$ for every such supercluster $\supercluster$.

Note that since in each phase we unmark all the vertices, unclustered vertices can become clustered again and in particular, edges between newly unclustered vertices and clustered vertices will be taken care of later on. This completes the description of the second step. See Fig. \ref{fig:lowexpansion} for an illustration.
\begin{figure}[h!]
\begin{center}
\includegraphics[scale=0.35]{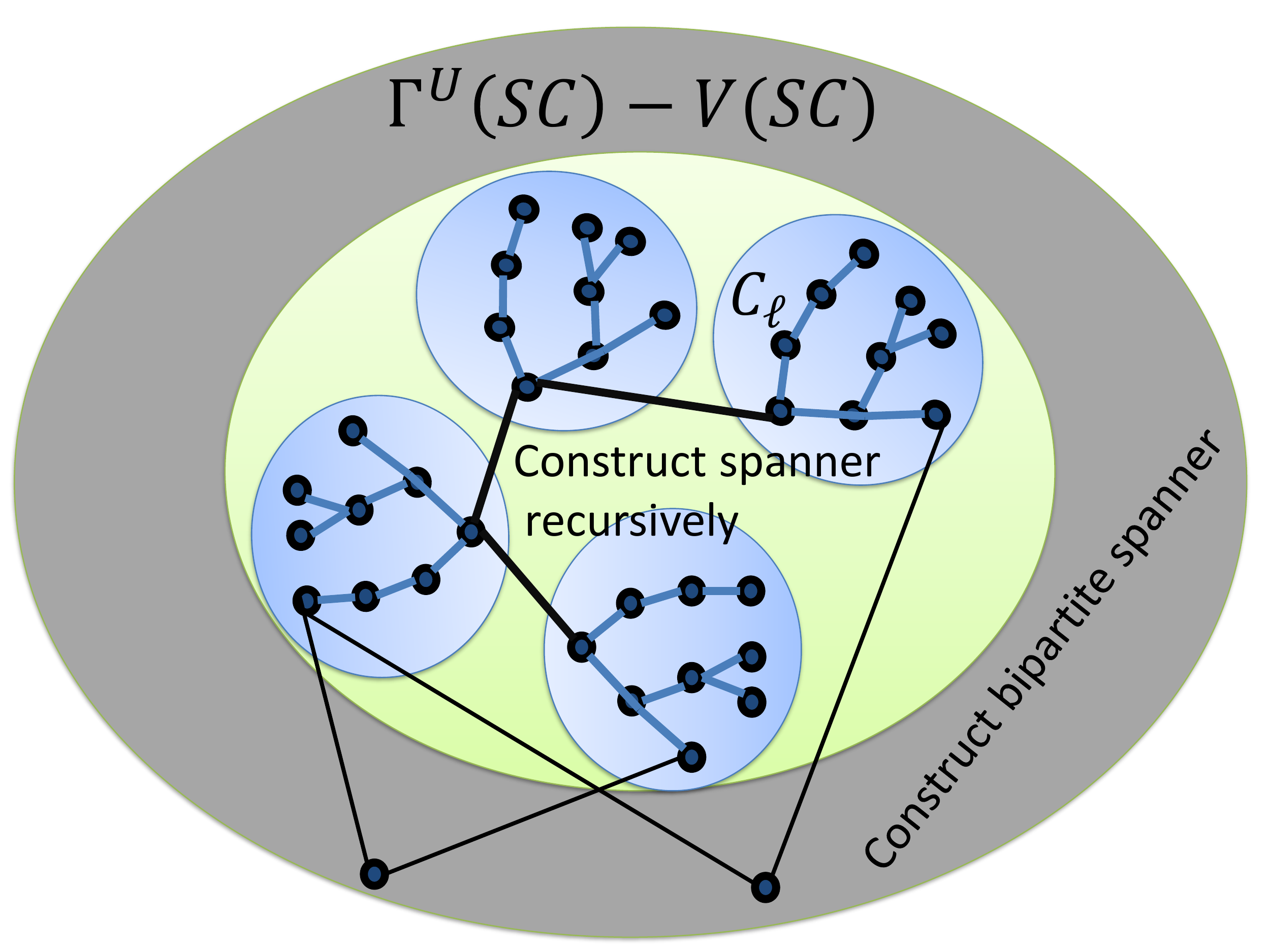}
\caption{
\label{fig:lowexpansion} Shown is supercluster $\supercluster$ of low-expansion. To provide a good stretch for the edges between its vertices -- we apply Alg. $\ImprovedSpanner$ recursively. To provide a good stretch between the supercluster vertices and its unmarked neighbors, we apply Alg. $\BipartiteSpanner$.}
\end{center}
\end{figure}

\paragraph{Steps (SIII) and (SIV): Defining $i^{th}$-clustering and the $i^{th}$-superclustering:}
The clusters $\mathcal{C}_{i}$ centered at the cluster centers $Z_i$ computed at step (SI) are computed exactly as in Baswana-Sen Algorithm. The depth $i$-trees of these clusters are added to the spanner. The main challenge here is to \emph{re-group} the new $O(n^{1-i/k})$ clusters into $O(\sqrt{n})$ superclusters, in a way that satisfies all the properties of the nice superclustering mentioned in Desc. \ref{desc:nice}. 

Our starting point is as follows: we have a collection of $O(n^{1/2-1/k})$ successful superclusters $\supercluster \in \superclustering_{i-1}$ whose cluster centers joined $Z_i$. Since $\superclustering_{i-1}$ is nice, by Prop. (N3), each such supercluster $\superclustering$ has a tree $T(\superclustering)$ of depth $O(2^{k})$ that spans all its cluster centers. 

First, we let each cluster $C \in \mathcal{C}_i$ with $\Omega(\sqrt{n})$ vertices, to define its own singleton superclusters. Since clusters are vertex disjoint, there are $O(\sqrt{n})$ such superclusters. It now remains to re-group the remaining clusters of $\mathcal{C}_i$ into $O(\sqrt{n})$ superclusters. 

For each successful supercluster $\supercluster$, we now consider only its centers of clusters with $O(\sqrt{n})$ vertices. First, we consider Property (N1) and use Lemma \ref{lem:balancedpartitioning} with the tree $T(\supercluster)$, weights $w(z) = 1$ for every cluster-center $z$ of $\supercluster$ (only those that have $O(\sqrt{n})$ vertices in their cluster) and bound $B = O(n^{1/2-i/k})$. All other vertices $v'$ in $T(\supercluster)$ have $w(v') = 0$ (in particular, the centers $z$ of clusters in $\supercluster$ which have been turned into singleton superclusters, we set $w(z) = 0$). 
By Prop. (N2) for $\superclustering_{i-1}$, we know that $\supercluster \in \superclustering_{i-1}$ has $O(n^{1/2-(i-1)/k})$ cluster centers. Hence, the partition procedure will partition each of these superclusters into $O(n^{1/k})$ superclusters $\supercluster_{1}, \ldots, \supercluster_{\ell}$. In addition, by Lemma \ref{lem:balancedpartitioning}(D5), all these resulting superclusters $\supercluster_{j}$ are equipped with edge-disjoint trees $T(\supercluster_j)$ of diameter $O(2^{k})$. 
Since there are $O(n^{1/2-1/k})$ successful superclusters, overall after this partition there are $O(n^{1/2-1/k})\cdot O(n^{1/k})$ superclusters. 

We then turn to property (N3), and farther partition the superclusters to obtain a balance partition of the \emph{vertices} into superclusters. For that purpose, for each supercluster $\supercluster'$ (obtained from the step above), we again apply the Partitioning Lemma on $T(\supercluster')$. This time we use $B=\sqrt{n}$ and the weight $w(z)$ of each cluster center $z$ in $\supercluster'$ is the number of vertices in its cluster $C$, that is $w(z)=|C|$ (for clusters $v'$ which have turned into singleton superclusters, or any other non-center vertex in $T(\supercluster')$, we set $w(v') = 0$). Since the vertices of superclusters are disjoint, this step increase the number of superclusters only by an additive $O(\sqrt{n})$ term. Hence, overall the number of superclusters is kept bounded by $O(\sqrt{n})$. This completes the description of the $i^{th}$ phase of Alg. $\ImprovedSpanner$.

\paragraph{The terminating step $k/2$.}
At the $(k/2)^{th}$ step we have $O(\sqrt{n})$ superclusters, each containing $O(1)$ clusters, hence overall we have $O(\sqrt{n})$ clusters. Now we can afford using Algorithm $\NaiveSpanner$ (described near the beginning of Section \ref{sec:kspanner}), which iterates over the \emph{clusters} one by one. This completes the description of the algorithm for graph with $\diam(G)=O(1)$. For a summary of the algorithm, see below.
\newpage
\begin{mdframed}[hidealllines=false,backgroundcolor=gray!30]
\centering  \textbf{\large  Phase $i$ of Algorithm $\ImprovedSpanner$} \\
\begin{flushleft}
\vspace{-5pt}
\textbf{(SI): Defining the centers $Z_i$.}\\
Set $\superclustering'_{i-1}\gets \superclustering_{i-1}$, $U=V$ and for $O(n^{1/2-1/k})$ steps do the following: 
\begin{itemize}
\item Each remaining supercluster $\supercluster_{i,j} \in \superclustering'_{i-1}$ computes $\deg^U(\supercluster_{i,j})$. 
\item The cluster-centers of local-maxima superclusters $\supercluster_{i,j} \in \superclustering'_{i-1}$ and with $\deg^U(\supercluster_{i,j}) \geq n^{1/2+1/k}$ join $Z_{i}$. We call these superclusters \emph{successful}. 
\item Remove the successful superclusters $\supercluster_{i-1,j}$ from $\superclustering'_{i-1}$ and mark the
vertices $\Gamma^U(\supercluster_{i-1,j})$.
\end{itemize}
\textbf{(SII): Taking care of unclustered clusters.}
At that point, all remaining superclusters $\supercluster_{i-1,j}\in \superclustering'_{i-1}$ have $\deg^U(\supercluster_{i,j})\leq n^{1/2+1/k}$. 
\begin{itemize}
\item For every singleton supercluster $\supercluster_{i-1,j}\in \superclustering'_{i-1}$ do the following: 
\begin{itemize}
\item 
For each unmarked vertex $u \in \Gamma^U(\supercluster_{i-1,j})$, add one edge $(u,v)$ to some $v \in V(\supercluster_{i-1,j})\cap \Gamma(u)$. 
\end{itemize}
\item Consider all remaining non-singleton superclusters $\supercluster_{i-1,j}\in \superclustering'_{i-1}$.
\item Let $\Gamma^{U,-}_{i,j}=\Gamma^U(\supercluster_{i-1,j})\setminus V(\supercluster_{i-1,j})$.
\item Let $B_{i,j}=(V(\supercluster_{i-1,j}),\Gamma^{U,-}_{i,j})$ be the corresponding bipartite graph.
\begin{itemize}
\item \textbf{Taking care of edges in $V(\supercluster_{i-1,j})\times \Gamma^{U,-}_{i,j}$:}
\begin{itemize}
\item Apply Alg. $\BipartiteSpanner$ on each $B_{i,j}$ to compute $(2k-1)$-spanners $H_{i,j}$ and add these $H_{i,j}$ spanners to $H$.
\end{itemize}
\item \textbf{Taking care of edges in $V(\supercluster_{i-1,j})\times V(\supercluster_{i-1,j})$:}
\begin{itemize}
\item Apply Algorithm $\ImprovedSpanner$ recursively on $G(V(\supercluster_{i-1,j}))$. 
\end{itemize}
\end{itemize}
\end{itemize}
\textbf{(SIII): Forming the $\mathcal{C}_i$ clusters centered at $Z_i$.} As in Baswana-Sen. Add the edges of the $i$-depth trees of each cluster to the spanner. \\
\textbf{(SIV):  Defining the $i^{th}$-superclustering $\superclustering_i$.} 
\begin{itemize}
\item Let each cluster $C \in \mathcal{C}_i$ with $\Omega(\sqrt{n})$ vertices form its own supercluster (i.e., singleton supercluster). 
\item For each \emph{successful} supercluster $\supercluster_{i-1,j}$ do the following:
\begin{itemize}
\item \textbf{Balance w.r.t. number of clusters:} Use the Partitioning Lemma \ref{lem:balancedpartitioning} to partition the superclusters into cluster-balanced superclusters -- each containing  $O(n^{1/2-i/k})$ (small) clusters. 
\item \textbf{Balance w.r.t. number of vertices:} Use the Partitioning Lemma \ref{lem:balancedpartitioning} again, on each of the resulting superclusters from above, to further partition them into vertex-balanced  superclusters, each with $O(\sqrt{n})$ vertices.  
\end{itemize}
\end{itemize}
\end{flushleft}
\end{mdframed}

\subsection{Analysis of $\ImprovedSpanner$ for Graphs with $\diam(G)=O(1)$}
\paragraph{Stretch Analysis.}
We begin by showing that the output subgraph is a $(2k-1)$ spanner.
A  vertex $v$ is $i$-clustered if it belongs to $V_i$, i.e., belongs to one of the clusters $C$ belonging to a supercluster $\supercluster_{i,j}$ of the $i^{th}$-superclustering $\superclustering_i$. Otherwise, it is $i$-unclustered.
Consider an edge $(u,v) \in G \setminus H$. Let $r(u)$ be the largest index such that $u$ is $r(u)$-clustered. The index $r(v)$ is defined similarly.  Without loss of generality, assume that $r(u)\leq r(v)$. 

First, consider the case where $r(v)+1 < k$, let $i=r(v)+1$. Hence, $v$ is $(i-1)$-clustered but both $u$ and $v$ are $i$-unclustered. Let $v$ be in cluster $C\in\supercluster_{i-1,j}$. Since at step $i$, $u$ is unclustered, it remains unmarked after computing all the superclusters with high-expansion at the $i^{th}$ phase. \\
\textbf{Case 1:} $u$ and $v$ belong to the same cluster $C$. Since $H$ contains a BFS tree of depth $i$ that spans all the vertices of $C$, we have that $\dist(u,v,H)\leq 2i-1$. \\
\textbf{Case 2:} $u$ and $v$ belong to different clusters in the same supercluster $\supercluster_{i-1,j}$. 
In such a case, by definition, $\supercluster_{i-1,j}$ is not singleton. 
Recall that the graph $G'_{i,j}$ is the induced subgraph on $V(\supercluster_{i-1,j})$, hence $(u,v) \in G'_{i,j}$. Since the algorithm computes a $(2k-1)$-spanner to the graph $G'_{i,j}$, the edge $(u,v)$ has a stretch of at most $2k-1$ in $H$. \\
\textbf{Case 3:} $u$ and $v$ belong to clusters of different superclusters.  
In this case, the edge $(u,v)$ is the in the bipartite graph $B_{i,j}=(\supercluster_{i-1,j}, \Gamma^{U,-}_i(\supercluster_{i-1,j}))$, where $\Gamma^{U,-}_i(\supercluster_{i-1,j})=\Gamma^{U}_i(\supercluster_{i-1,j})\setminus V(\supercluster_{i-1,j})$. 
Since we add to $H$ a $(2k-1)$-spanner for $B_{i,j}$, the edge $(u,v)$ has a stretch of at most $(2k-1)$ in $H$.\\
\textbf{Case 4:} $u$ unclustered.
In this case, the edge $(u,v)$ is the in the bipartite graph $$B_{i,j}=(\supercluster_{i-1,j}, \Gamma^{U,-}_i(\supercluster_{i-1,j})).$$ Since we add to $H$ a $(2k-1)$-spanner for $B_{i,j}$, the edge $(u,v)$ has a stretch of at most $(2k-1)$ in $H$.

Finally, consider the case where $r(v)=k-1$. In the last stage the cluster containing $v$ adds an edge to each neighboring vertex, including $u$. Hence, since the radius of the cluster of $u$ is $k-1$, the distance in $H$ from $u$ to $v$ is at most $2k-1$, as desired.
%
%
%
%
%



\paragraph{Size Analysis.}
We proceed by showing that the output subgraph contains $O(k \cdot n^{1+1/k})$ edges.
At the step of computing the $i^{th}$-superclustering, we add the following edges to the spanner $H$. For the clusters in a supercluster that has high-expansion, we add BFS trees of depth $i$ rooted at their centers. Since the clusters are vertex-disjoint, in total we add $O(n)$ edge to the spanner. For each supercluster $\supercluster_{i-1,j}$ of low-expansion, we first use Lemma \ref{bipartite} to construct a $(2k-1)$-spanner for the bipartite graph 
$B_{i,j}=(\supercluster_{i-1,j}, \Gamma^{U,-}(\supercluster_{i-1,j}))$ with $O(k \cdot n^{1/2+1/k})$ edges. Since there are $O(\sqrt{n})$ superclusters in $\superclustering_{i}$, this adds $O(kn^{1+1/k})$ edges to the spanner $H$. 

Next, we bound the number of edges added in Step (SII) of phase $i$, due to the computation of the $(2k-1)$-spanners for each of the induced subgraphs $G'_{i,j}$ on the vertices of $V(\supercluster_{i-1,j})$ for every supercluster $\supercluster_{i-1,j}\in \superclustering_i$ of low-expansion. Since this is done only for non-singleton superclusters, by Prop. (N2) of the nice superclustering, $|V(\supercluster_{i-1,j})|=O(\sqrt{n})$. Since there are $O(\sqrt{n})$ superclusters in the superclustering $\superclustering_{i-1}$, in total it adds $O(\sqrt{n})\cdot O((\sqrt{n})^{1+1/k})=O(n^{1+1/2k})$ edges. 

After $k/2$ phases of Alg. $\ImprovedSpanner$, we have $O(\sqrt{n})$ clusters and at that point we apply Alg. $\NaiveSpanner$ on these remaining clusters (i.e., it is like starting Alg. $\NaiveSpanner$ from $(k/2)^{th}$ phase). This adds only $O(k n^{1+1/k})$ edges to the spanner by using the exact same argument as in Sec. \ref{sec:naive}. Finally, in the terminating part, where we only have $O(n^{1/k})$ clusters, we add $O(n^{1/k})$ edges for each vertex (i.e., each vertex $v$ adds an edge to one neighbor $u \in C \cap \Gamma(v)$ for every cluster $C$), for a total of $O(n^{1 + 1/k})$ edges as well. This completes the size analysis of the spanner.

\paragraph{Round Complexity.}
By Lemma \ref{lem:zerogeneral}, constructing the $0^{th}$-superclustering $\superclustering_0$ takes $O(2^{k} \cdot n^{1/2 - 1/k})$ rounds. 
We now claim that for $i \geq 1$, constructing the $i^{th}$-superclustering from the $(i-1)^{th}$-superclustering takes $O(2^{k} \cdot n^{1/2 - 1/k})$ rounds. By property (N3), the centers of all clusters in supercluster $\supercluster_{i-1,j}$ are connected by subgraph of depth $O(2^{k})$. Since these subgraphs are edge-disjoint, the unmarked degree of all the superclusters $\supercluster_{i-1,j} \in \superclustering_{i-1}$ can be computed in $O(i \cdot 2^{k})$ rounds.

We now consider step (S2).  By Lemma \ref{bipartite}, constructing a $(2k-1)$-spanners for a bipartite graphs $B_{i,j}=(V(\supercluster_{i-1,j}), \Gamma_i^{U,-}(\supercluster_{i-1,j}))$ takes $O(k \cdot n^{1/2-1/k})$ rounds, and since each edge is contained in at most two of these subgraphs, all these spanners can be constructed in parallel within $O(k \cdot n^{1/2 - 1/k})$ rounds (while slowing down the algorithm of Lemma \ref{bipartite} by a factor of 2, due to the congestion).

Next, the algorithm considers all the non-singleton superclusters $\supercluster_{i-1,j}$ with low-expansion and recursively computes a $(2k-1)$-spanner for the subgraph $G'_{i,j}$ (the induced subgraph on the vertices of 
$\supercluster_{i-1,j}$). Since the superclusters are vertex disjoint and this is done only for non-singleton superclusters $\supercluster_{i-1,j}$ which by Prop. (N2) consists of $O(\sqrt{n})$ vertices, it takes $O((\sqrt{n})^{1/2-1/k})=O(n^{1/4})$ rounds, for all these superclusters in parallel.

After $k/2$ phases, we have $\sqrt{n}$ clusters and at that point, applying Alg. $\NaiveSpanner$ for $k/2-1$ phases takes $O(k \cdot n^{1/2-1/k})$ rounds. The terminating part, where we have $O(n^{1/k})$ clusters, and each vertex adds an edge to a unique neighbor in each these clusters takes $O(1)$ rounds.

\paragraph{Extension for general graphs of diameter $\diam(G)$.}
The only step that requires adaptation is that of constructing the $0^{th}$-level superclustering $\superclustering_0$. In Appendix \ref{append:largediameter}, we show:
\begin{lemma}
\label{lem:zerogeneral}
One can construct in $O(2^{k} \cdot n^{1/2-1/k})$ rounds, nice superclustering $\superclustering_0=\{\supercluster_{0,1}, \ldots, 
\supercluster_{0,p}\}$ along with a subgraph $H'$ with $O(kn^{1+ 1/k})$ edges such that for every vertex $u$ not participating in the clusters of these superclusters $\supercluster_{0,j}\in \superclustering_0$, it holds that $\dist(u,v,H')\leq 2k-1$ for every $v \in \Gamma(u)$.
\end{lemma}
\bibliographystyle{plain}
\bibliography{bibfile}

\newpage
\appendix

\section{Partitioning Lemma \ref{lem:balancedpartitioning}}\label{append:partition}
\APPENDBALANCEPART



\section{Missing Details of Alg. $\ImprovedSpanner$ for General Graphs}\label{append:largediameter}
In this section, we prove Lemma \ref{lem:zerogeneral}. This is done by constructing $0^{th}$-Superclustering for graphs with arbitrary diameter.
\subsection*{Constructing the $0^{th}$-Superclustering} 
Our goal in this section is to construct a $0^{th}$-level nice superclustering. For the properties of the nice superclustering, please see Description \ref{desc:nice}. In contrast to our algorithm for $O(1)$-diameter graphs, the $0$-superclustering $\supercluster_0$ would not cover all the vertices in the graph. However, for all vertices $V'$ that are uncovered by the $0$-superclustering, we will add edges to the spanner $H$ so that $\dist(u,v,H)\leq V'$ for every $u \in V'$ and $v \in \Gamma(u)$. 

We make use of the following procedure that computes ruling set in power graphs.

\paragraph{Tool: computing ruling-set in the power-graph.}
Our algorithm is based on the algorithm \textsc{Generic\_Spanner} from \cite{derbel2007deterministic}. 
It uses the $(3, 4)$ ruling set algorithm from \cite{panconesi1992improved}, slightly adapted to our case to find a $(3, 4)$ ruling set for the graph $G^t$, while operating in the \congest\ model on the graph $G$ (i.e., we find a $(3t, 4t)$-ruling set on $G$):
\begin{lemma}\label{rulingset}
Given a graph $G = (V, E)$, a subset of vertices $V' \subseteq V$ (each vertex knows whether it is in $V'$ or not), and a positive integer parameter $t$, one can construct a subset $X \subseteq V'$ that is a 
$(3, 4)$-ruling set with respect to $G^{t}$ and $V'$. Furthermore, the algorithm runs in $t2^{O(\sqrt{\log n})}$ rounds in the \congest\ model.
\end{lemma}
\begin{proof}
The algorithm is based on the $(3, 4)$-ruling set algorithm of \cite{panconesi1992improved}. We simply adapt their algorithm to work on the power-graph $G^t$ instead of $G$.
We will go through all of the steps of the algorithm, and confirm the algorithm for $G^t$ can be simulated in $G$ using $O(\log n)$-size messages.
First, every vertex $v$ in parallel must be able to test whether $d(v) \ge g(n)$, where $d$ refers to the degree in $G^t$.
This can be done as follows. Each vertex $v$ sends to each of its neighbors $u$, the list of its neighbors $\Gamma(v)$, if $|\Gamma(v)| < g(n)$. Otherwise, $v$ sends a message specifying that $|\Gamma(v)| \ge g(n)$.

Thus, each vertex can tell whether it has high degree in $G^2$. If it does not have high degree, it has a list of all its neighbors in $G^2$. It is simple to extend this algorithm for larger $t$ by having $t$ steps where on the $i^{th}$ step, each vertex sends the list $\Gamma_i(v)$ of all vertices at distance at most $i$ (or a message specifying that $|\Gamma_i(v)| \ge g(n)$). Hence, the size of the messages sent is bounded by $g(n)$. Overall, the algorithm has $t$ steps, each with messages of size $g(n)$, which can be simulated using short messages in a total of $t\cdot g(n)$ rounds.

Next, we must check whether we can compute a $g(n)$-coloring on a graph where all vertices have degree less than $g(n)$. This can be done using the standard coloring algorithm: partition the vertices into two parts based on the last bit of their ID, then color each part, and merge. To merge, we simply have each vertex figure out the set of colors at distance $t$ away from it, which is possible with $t$ rounds with messages of size $g(n)$.

Next, we must verify that we can find a $(3, 3\log n)$-ruling set. For this we use the standard $(3t, 3t \log n)$ ruling set algorithm from \cite{panconesi1992improved}.
Since all of the steps of the algorithm for $(3, 4)$-ruling set on $G^t$ can be simulated in $G$ in $tg(n)$ rounds, we can construct a $(3t, 4t)$-ruling set in $t2^{O(\sqrt{\log n})}$ rounds.
\end{proof}
We are now ready to describe Alg. $\ConstructZeroSuper$, the output of the algorithm is a nice $\superclustering_0$ along with a partial spanner $H'$. 

\subsection*{Description of Algorithm $\ConstructZeroSuper$}
To construct the $0^{th}$-superclustering, the algorithm computes $k/2$ levels of pre-clustering $\Zclustering_{0},\ldots, \Zclustering_{k/2}$. 
Each clustering $\Zclustering_{i}$ contains $O(n^{1-i/k})$ clusters, but 
unlike the clustering of Baswana-Sen, the diameter of each cluster $C \in  \Zclustering_{i}$ is exponential in $i$ (i.e., in Baswan-Sen, the diameter is $i$). Again, we denote the cluster centers of level $i$, by $Z_i$, and by $V_i$ the set of $i$-clustered vertices.

After $k/2$ phases, the $(k/2)^{th}$ clustering $\Zclustering_{k/2}$ contains $O(\sqrt{n})$ vertex-disjoint clusters, for each cluster $C$, we have a tree $T(C)$ of diameter $O(2^{k})$ that spans all the vertices in $C$. We then turn these $O(2^{k})$-diameter clusters into superclusters in $\superclustering_0$. 


We begin by describing phase $i$ for $i \in \{1, \ldots, k/2\}$. The following definitions are useful.
For every supercluster $C \in \Zclustering_{i-1}$, define:
\begin{equation}
\label{eq:zeroexpans}
\Gamma^+(C)=C \cup \bigcup_{u \in C}\Gamma(u) \mbox{~~and~~} \deg(C)=|\Gamma^+(C)|
\end{equation}
We say that  $C \in \Zclustering_{i-1}$ has \emph{high-expansion} if $\deg(C)\geq n^{i/k}$. Otherwise, it has \emph{low-expansion}. 

The initial pre-clustering $\Zclustering_{0}=\{\{v\}\}$ consists of $n$ singleton clusters, each spanning a subtree of diameter $0$. For $\Zclustering_{i}$, the set $Z_i$ is the set of cluster centers of the clusters of $\Zclustering_{i}$. As we will see, for every $i$,  $\Zclustering_{i}$ consists of $O(n^{1-i/k})$ clusters. Each cluster $C_j \in \Zclustering_{i}$ has a center $z_i \in Z$ and a tree of diameter $2^{4i}$ rooted at the cluster center spanning all the cluster vertices. Hence, our clustering staifies all the properties (P1-P3) of Baswana-Sen, with the only distinction of having diameter $2^{4i}$ rather than $i$. As will be clear later, these clusters are not important not for the purpose of having small stretch, but rather to create the initial supercluering (hence, the factor $2^{4k}$ appears in our round complexity -- the cost of communicating inside superclusters, but not in the stretch bound). 
Each phase $i\in \{1, \ldots, k/2\}$, consists of 3 steps as defined next. 

\paragraph{Step (S1): Selecting the cluster centers $Z_i \subseteq Z_{i-1}$:}
Every cluster $C \in  \Zclustering_{i-1}$ computes its $\deg(C)$ in $O(2^{4i})$ rounds. 
A cluster $C$ with high-expansion, adds its center to $Z'_{i-1}$. Hence, $Z'_{i-1}$ is the cluster centers of all clusters with high-expansion (i.e., $\deg(C)\geq n^{i/k})$). 

Next, we use Lemma \ref{rulingset} to compute  $(2^{4i}, 2^{4(i-1)})$ ruling set $Z_{i} \subseteq Z'_{i-1}$ in $G$. This can be done in $O(2^{\sqrt{\log n}+4i})$ rounds using Lemma \ref{rulingset} with $V'=Z_{i-1}$. Since all the cluster centers that belong to $Z_i$ have a cluster with high expansion, the number of clusters $|Z_{i}|$ is bounded by $n/n^{i/k}=n^{1-i/k}$ (this is shown later explicitly in the analysis). 

\paragraph{Step (S2): Taking Care of unclustered vertices $V_{i-1}\setminus V_i$:}
All the newly unclustered vertices belong to low-expanded clusters $C \in \Zclustering_{i}$, i.e., with $\deg(C)\leq n^{i/k}$. We first consider all edges connecting $C$ to vertices in $V \setminus C$. For every such cluster $C$, we define the bipartite graph $B(C)=(C, \Gamma(C)\setminus C)$ and apply Alg. $\BipartiteSpanner$ to construct a $(2k-1)$-spanner $H(C) \subseteq B(C)$. In the analysis part, we show that this step adds $O(n^{1+1/k})$ edges to the spanner, and can be implemented in $O(n^{1/2-1/k})$ rounds.

Finally, it remains take care of the edges internal to the low-expanded clusters $C$. By definition, the low-exapnded $C$ has less than $n^{i/k}$ vertices, and since $i\leq k/2$, it has $O(\sqrt{n})$ vertices. Since the clusters are vertex disjoint, we simply apply the algorithm recursively on the induced graph $G(C)$ of each such cluster $C$. 

\paragraph{Step (S3): Defining the clustering $\mathcal{C}_i$:}
Finally, we use the cluster centers $Z_i$ to define the new clusters of $\Zclustering_{i}$. This step is done exactly as in Baswana-Sen, only that in our case, the diameter of each cluster is  $2^{4(i-1)}$. A summary description of phase $i$ is given below.
\newpage
\begin{mdframed}[hidealllines=false,backgroundcolor=gray!30]
\centering  \textbf{\large  Phase $i$ of Algorithm $\ConstructZeroSuper$} \\
\begin{flushleft}
\vspace{-5pt}
\textbf{(SI): Defining the centers $Z_i$.}\\
\begin{itemize}
\item
Let $Z'_{i-1}$ be the cluster centers of all clusters $C \in \Zclustering_{i-1}$ with $\deg(C)\geq n^{i/k}$.
\item 
For the $Z'_{i-1}$ vertices, compute $(2^{4i}, 2^{4(i-1)})$ ruling set $Z_{i} \subseteq Z'_{i-1}$
\end{itemize}
\textbf{(SII): Taking care of unclustered vertices.}\\
For every $C_j \in \Zclustering_{i-1}$ with $\deg(C_j)\leq n^{i/k}$, do the following:
\begin{itemize}
\item
\textbf{Taking Care of Edges in $C_j \times (\Gamma(C_j)\setminus C_j)$}:\\
\begin{itemize}
\item
Apply Alg. $\BipartiteSpanner$ on the bipartite graph $B_{i-1,j}=(C_j,(\Gamma(C_j)\setminus C_j))$.
\item
Add the $(2k-1)$-spanner $H_{i-1,j} \subseteq B_{i-1,j}$ to $H$. 
\end{itemize}
\item 
\textbf{Taking Care of Edges in $C_j \times C_j$:}\\
Apply Alg. $\ImprovedSpanner$ recursively on the graph $G(C)$.
\end{itemize}
\textbf{(SIII): Forming the $\Zclustering_{i}$ clusters centered at $Z_i$.} As in Baswana-Sen only that this time, the diameter of each cluster is $2^{4(i-1)}$. 
\end{flushleft}
\end{mdframed}
\paragraph{Termination Phase:}
After $k/2$ phases, the $(k/2)^{th}$ clustering $\Zclustering_{k/2}$ contains $O(\sqrt{n})$ vertex-disjoint clusters, for each cluster $C$, we have a tree $T(C)$ of diameter $O(2^{4k})$ that spans all the vertices in $C$. We then turn these $O(2^{4k})$-diameter clusters into superclusters in $\superclustering_0$. To do that, we need the $O(\sqrt{n})$ clusters to be balanced with respect to the number of vertices they contain. Hence, we apply the Partitioning Lemma \ref{lem:balancedpartitioning} to turn to $O(\sqrt{n})$ clusters of $\Zclustering_{k/2}$ into $O(\sqrt{n})$ clusters each with $O(\sqrt{n})$ vertices. Specifically, on each $C \in \Zclustering_{k/2}$ with at least $\sqrt{n}$ vertices, apply the Partitioning Lemma \ref{lem:balancedpartitioning} on $T(C)$, $w(v)=1$ for every $v \in C$ and bound $B=\sqrt{n}$. By that we get $O(\sqrt{n})$ clusters, each with $\ell=O(\sqrt{n})$ vertices. At that point, each such cluster is now defined as a supercluster in $\superclustering_0$ with $O(\sqrt{n})$ singelton clusters.
\begin{lemma}
\label{lemma:numclusters}
For every $i \in \{0,\ldots, k/2\}$, $\Zclustering_{i}$ contains $O(n^{1-i/k})$ clusters. 
\end{lemma}
\begin{proof}
The lemma follows from the following two facts:
\begin{enumerate}
\item
The clusters are disjoint.
\item
Each cluster contains at least $n^{i/k}$ vertices.
\end{enumerate}
The first item follows by construction (in step (SIII), when constructing the clusters, we break ties based on IDs to ensure that every vertex joins at most one cluster).

The second item follows by noting that if a center center $z$ joins $Z_i$, its cluster in $\Zclustering_{i-1}$ had degree at least $n^{i/k}$. Also, no other cluster which joined $Z_i$ neighbored any of those $n^{i/k}$, since the centers joining $Z_i$ form a ruling set of $Z_{i-1}$, Hence, these $n^{i/k}$ neighbors will join the cluster centered at $z$ in the clustering $\Zclustering_{i-1}$.
\end{proof}

\paragraph{Size Analysis of Step (SII):}
By Lemma \ref{lemma:numclusters}, at step $i$, we have $O(n^{1-(i-1)/k})$ clusters in $\Zclustering_{i-1}$. 
By definition of the cluster of low-expansion, $\deg(C_{i-1,j}) \leq n^{i/k}$. We then construct a $(2k-1)$
spanner $H_{i-1,j}$ for the bipartite graph $B_{i-1,j}=(C_{i-1,j}, \Gamma(C_{i-1,j})\setminus C_{i-1,j})$. 
By Lemma \ref{bipartite}, $E(H_{i-1,j})$ can be bounded by: 
\begin{eqnarray*}
E(H_{i-1,j})&=&O(|C_{i-1,j}|^{1+2/k}+|\Gamma(C_{i-1,j})\setminus C_{i-1,j}|)
\\&=&
O(k\cdot |C_{i-1,j}|^{1+2/k}+n^{i/k})\leq O(k \cdot|C_{i-1,j}|\cdot n^{1/k}+n^{i/k}).
\end{eqnarray*}
Since all clusters are vertex disjoint, $|C_{i-1,j}| \leq n^{i/k}\leq \sqrt{n}$, and there $O(n^{1-(i-1)/k})$ clusters, we have that over all 
we added $O(k \cdot n^{1+1/k})$ edges to the spanner. In fact, a tighter analysis can give only $O(n^{1+1/k})$ edges.

Finally, we claim that the edges added by applying the algorithm recursively on a cluster $C$ can be bounded by $O(k \cdot |C|^{1+1/k})$ as well. This follows since clusters are vertex disjoint, and in addition, each cluster is of size at most $n^{i/k}$, thus $O(n^{1 + i/k^2})$ edges added to the spanner in the $i^{th}$ phase. In total, we have $O(n^{1 + 1/k})$ edges added to the spanner.

\paragraph{Round Analysis of Step (SII):} It is easy to see that each edge in the bipartite graph $B_{i_1,j}$ appears in at most $2$ bipartite graphs. Hence, when applying Lemma \ref{bipartite} on all of the bipartite graphs in parallel, the algorithm will experience congestion of at most $2$. Then, by Lemma \ref{bipartite}, it will take at most $O(kn^{1/2 + 1/k})$ rounds, since each of the bipartite graphs has at most $n^{1/2}$ vertices on one side (the side corresponding to the cluster $C$).

When constructing the spanner recursively, each of the clusters in which we recursively construct a spanner is of size up to $\sqrt{n}$, and can all be done in parallel (since the clusters are disjoint). Hence, the step of the recursion will take a total of $T(\sqrt{n})$ rounds (where $T(n)$ denotes the number of rounds used by the algorithm on an $n$-vertex graph), which is negligible compared to $T(n)$ (even if run $k/2$ times, one for each phase).

\begin{lemma}\label{lem:0superclusteringrunningtime}
Algorithm $\ConstructZeroSuper$ terminates within $O(2^{4k}n^{1/2 - 1/k})$ rounds.
\end{lemma}
\begin{proof}
For Step (SI), we first have each cluster compute its degree. This can be done in $O(\diam(C))$ rounds in a way similar to that of algorithm $\NaiveSpanner$: ever vertex sends a message to all its neighbors containing its cluster ID, and then each vertex sends one ACK to each neighboring cluster. The cluster can then upcast the count of number of ACKs it received, which is exactly its degree.
The next stage of Step (SI) takes the same number of rounds as computing the $(2^{4i}, 2^{4(i-1)})$ ruling set $Z_{i} \subseteq Z'_{i-1}$, which by Lemma \ref{rulingset} takes $O(2^{4i}2^{O(\sqrt{\log n})})$ rounds. Summing over all $k/2$ phases gives $O(2^{4k}2^{O(\sqrt{\log n})})$ rounds.
Step (SII) takes $O(kn^{1/2 + 1/k})$ rounds, as shown above.
Step (SIII) involves running BFSs of depth $O(2^{2i})$ from each of the cluster centers in $Z_i$ (as well as breaking ties based on IDs). This takes a total of $O(2^{2i})$ rounds. Summing over all $k/2$ phases gives $O(2^k) \le O(2^{4k})$ rounds.
%
\end{proof}
We are now ready to complete the proof of Lemma \ref{lem:zerogeneral}. 
\begin{proof}[Lemma \ref{lem:zerogeneral}]
By Lemma \ref{lem:0superclusteringrunningtime}, the  algorithm takes $O(2^{4k}\cdot n^{1/2 - 1/k})$ rounds.
By the size analysis above, we also have the total number of edges added to the spanner is $O(kn^{1+1/k})$.

\paragraph{Showing that the output $0^{th}$-superclustering $\superclustering_0$ is nice.}
We first show that $\Zclustering_{i}$ contains $O(n^{1-i/k})$ clusters. This follows since the set $Z_i\subseteq Z'_{i-1}$ is a $(2^{4i}, 2^4(i-1))$ ruling-set, where each $z \in Z'_i$ have $\Omega(n^{i/k})$ vertices in its $(2^{4(i-2)}+1)$-neighborhood. Since the $(2^{4(i-2)}+1)$-neighborhoods of $z,z' \in Z_i$ are disjoint, $Z_i=O(|Z_{i-1}|/n^{i/k})=O(n^{1-i/k})$. 

Therefore after $k/2$ stages,  $\Zclustering_{k/2}$ contains $O(n^{1/2})$ clusters. We now show that these clusters can define the superclusters of $\superclustering_0$ while satisfying all the nice properties of Desc. \ref{desc:nice}. Each of these clusters $C$ has a tree $T(C)$ of diameter at most $2^{4k}$ (since in each phase, the radius of the cluster increases by a constant factor less than 16). For every cluster $C_i \in \Zclustering_{k/2}$ with $\Omega(\sqrt{n})$ vertices, we can  use the Partitioning Lemma \ref{lem:balancedpartitioning} on the tree $T(C_i)$, weights $w(v)=1$ for each $v \in T(C_i)$ and bound $B=\sqrt{n}$. This will partition $C_i$ into a collection $C_{i_1}, \ldots, C_{i_\ell}$ of new clusters (along with a subtree in $T(C)$ for each $C_{i_j}$) such that $|C_{i_j}|=\Theta(n)$ for every $i\geq 2$ and $|C_{i,1}|=O(n)$, i.e., there exists at most one cluster, $C_{i_1}$, with less than $\sqrt{n}$ vertices. Since the clusters of $\Zclustering_{k/2}$ are vertex disjoint, applying that to all clusters $\Zclustering_{k/2}$ with at least $n^{1/2}$ clusters, partitions these clusters into a a total of $O(\sqrt{n})$ new clusters each with $O(\sqrt{n})$ vertices. The partitioning lemma also gives a tree $T(C_{i_j}) \subseteq T(C_i)$ that connects the vertices of $C_{i_j}$. Note that since the $C_i$ clusters in $\Zclustering_{k/2}$ are vertex disjoint, the trees $T(C_{i_j})$ of all new clusters $C_{i_j}$ are edge-disjoint. More accurately, by (D4) of Lemma \ref{lem:balancedpartitioning}, $T(C_{i_{j}}),T(C_{i_{j'}})$ (trees of new clusters that originate at the same cluster $C_i$) are edge-disjoint for $j\neq j'$; and the trees of $T(C_{{i}_{j}}),T(C_{{i'}_{j'}})$ for $i\neq i'$ (trees of new clusters that originate at the different clusters $C_i,C_{i'}$) are vertex disjoint, since $C_i,C_{i'}$ are vertex disjoint and $T(C_{{i}_{j}})\subseteq T(C_i)$. 

Finally, each of the $O(\sqrt{n})$ new-clusters $C_{i_j}=\{v_{i_j,1}, \ldots, v_{i_j,\ell}\}$ with $\ell=O(\sqrt{n})$ vertices becomes a supercluster $\supercluster_{0,i_j}=\{\{v_{i_j,1}\}, \ldots, \{v_{i_j,\ell}\}\}$, i.e., each vertex $v_{i_j,p}$ is singelton cluster $\{v_{i_j,p}\}$ of $\supercluster_{0,i_j}$. All properties of Desc. \ref{desc:nice} hold. 

\paragraph{Showing that all edges of unclustered vertices are handled in $H$.}
Consider an edge $(u, v)$ where either $u$ or $v$ is $k/2$-unclustered (i.e., not appearing in the clusters of $\superclustering_0$). We will show that the edge has a stretch of at most $(2k - 1)$ in the partial spanner $H$. Consider the first phase $i$ in which either $u$ or $v$ became unclustered. Then in the previous phase $(i-1)$, both $u$ and $v$ are clustered. Suppose without loss of generality that $u$ becomes unclustered in phase $i$. There are two cases:\\
\textbf{Case 1:} $u$ and $v$ belong to the same cluster $C \in \Zclustering_{i-1}$. Since we apply Alg. $\ImprovedSpanner$ recursively in the cluster $C$  (see Step (SII)) and add the output spanner for $H$, we have $\dist(u,v,H)\leq 2k-1$. \\
\textbf{Case 2:} $u$ and $v$ belong to different clusters in $\Zclustering_{i-1}$: 
In such a case, the center of $v$'s cluster did not join $Z_i$ (otherwise $u$ would be clustered). Since we employ Alg. $\BipartiteSpanner$ on a bipartite graph between $u$'s cluster and its neighbors (which include $v$), we have $\dist(u,v,H)\leq 2k-1$ (see Step (SII)).
\end{proof}
For the case of odd $k$, we use a similar algorithm and analysis. The only difference is that we terminate our for loop at $(k-1)/2$ instead of $k/2$, and use the bipartite spanner algorithm for odd $k$. This yields the following similar lemma:
%
%

\begin{lemma}[$0$-Superclustering Lemma: Odd $k$]\label{partition}
Let $G = (V, E)$ be a graph and $k \ge 3$ be a positive odd integer. One can construct the nice $0^{th}$-level superclustering in $O(n^{1/2 - 1/(2k)})$ rounds.
\end{lemma}
%
%
%
%
%
%
%
%
%
%
%
%

\end{document}